\documentclass[10pt,twocolumn,twoside]{IEEEtran}
\usepackage[left=0.75in,right=0.75in,top=0.73in,bottom=.76in]{geometry}
\usepackage[utf8]{inputenc}
\usepackage[T1]{fontenc}
\usepackage{url}
\usepackage{ifthen}
\usepackage{cite}

\usepackage{graphicx,epstopdf,amssymb,amsthm}% ,epstopdf,spconf,epsfig,fullpage}%, ,graphicx,psfrag,url,amsmath,amsthm,comment}
\usepackage{ulem,color}
\usepackage{epstopdf}
\usepackage[utf8]{inputenc}
\usepackage[english]{babel}
\usepackage{caption}
\usepackage[dvipsnames]{xcolor}

\newcommand{\bs}[1]{\textcolor{black}{#1}}

\newcommand{\R}{{\rm I\!R}}

\newtheorem{theorem}{Theorem}
\newtheorem{lemma}{Lemma}
\newtheorem{proposition}{Proposition}
\newtheorem{corollary}{Corollary}
\newtheorem{assumption}{Assumption}
\newtheorem{remark}{Remark}

\newcommand{\diag}{\text{diag}}

\newtheorem{definition}{Definition}
\usepackage{cite}
\usepackage[T1]{fontenc}
\usepackage{color}
\usepackage{amsmath}
\usepackage{amsthm}
\usepackage{amssymb}
\usepackage{stmaryrd}
\usepackage{graphicx}
\usepackage{esint}

\makeatletter
\begin{document}

%Here goes the title

\title{Peak Infection Time for a Networked SIR Epidemic
% Coupled 
with Opinion Dynamics}

\author{%\IEEEauthorblockN{
Baike She, %\IEEEauthorrefmark{2}, 
Humphrey C. H. Leung, 
Shreyas Sundaram, %\IEEEauthorrefmark{1},
and %\IEEEauthorrefmark{2},
Philip E. Par\'{e}*%\IEEEauthorrefmark{2}
\thanks{*Baike She, Humphrey C. H. Leung, Shreyas Sundaram, and Philip E. Par\'{e} are with the Elmore Family School of Electrical and Computer Engineering at Purdue University.
E-mails: \{bshe, leung61, sundara2, philpare\}@purdue.edu.
Research supported in part by the C3.ai Digital Transformation Institute sponsored by C3.ai Inc. and the Microsoft Corporation, and in part 
   by the National Science Foundation, grants NSF-CMMI \#1635014, 
  NSF-CNS \#2028738, 
   and NSF-ECCS \#2032258.
}
%Eldon Tyrell\IEEEauthorrefmark{4},~\IEEEmembership{Fellow,~IEEE}}

%\IEEEauthorblockA{\IEEEauthorrefmark{1}
%  Department of Electrical and Computer Engineering, Purdue University, West Lafayette, IN, 47907 USA
% \IEEEauthorblockA{\IEEEauthorrefmark{2} School of Electrical and Computer Engineering, Purdue University, West Lafayette, IN, 47907, USA}
%\IEEEauthorblockA{\IEEEauthorrefmark{3}Starfleet Academy, San Francisco, CA 96678 USA}
%\IEEEauthorblockA{\IEEEauthorrefmark{4}Tyrell Inc., 123 Replicant Street, Los Angeles, CA 90210 USA}% <-this % stops an unwanted space
%\thanks{Manuscript received December 1, 2012; revised August 26, 2015. 
%Corresponding author: M. Shell (email: http://www.michaelshell.org/contact.html).}}
}

%\author
%{\IEEEauthorblockN{Author 1}
%\IEEEauthorblockA{School of Electrical and\\Computer Engineering\\
%University\\
%Location\\
%Email: }
%\and
%\IEEEauthorblockN{Author 2}
%\IEEEauthorblockA{University\\
%Location\\
%Email: }
%}
\maketitle

%Main body starts

\begin{abstract}
We propose an SIR epidemic model coupled with opinion dynamics to study an epidemic and opinions spreading in a network of communities. Our model couples networked SIR epidemic dynamics with opinions towards the severity of the epidemic, and vice versa. We develop an epidemic-opinion based threshold condition to capture the moment when a weighted average of the epidemic states starts to decrease exponentially fast over the network, namely the peak infection time. We define an effective reproduction number to characterize the behavior of the model through 
% sufficient conditions on 
%\phil{[I commented out some stuff here to make it flow better and I think it's still true]}
%upper and lower bounds for 
the peak infection time. %as a function of the susceptible proportions. 
We use both analytical and simulation-based results to illustrate that 
% the behavior of 
the opinions 
% can 
reflect the recovered levels within the communities after the epidemic dies out.
% towards the seriousness of the epidemic over the network after the epidemic disappears. 
\end{abstract}

%\begin {IEEEkeywords}
%IoT, Ontology, Semantics,  SSN, OWL, OBOE, OpenIoT, SWEET, SUMO
%\end{IEEEkeywords}

\section{Introduction}
\label{intro}
The COVID-19 pandemic has caused severe suffering across the world in both public health and the economy.
These hardships have motivated researchers from various backgrounds to study the viral pathogenesis, epidemic spreading models, mitigation strategies\cite{cao2020covid,anderson2020will}, etc. 
%Scientists and engineers from different research background are studying the epidemic spreading from different perspectives, spanning from the viral pathogenesis\cite{cao2020covid}, modeling the spreading processes\cite{ihme2020modeling}, the mitigation strategies\cite{anderson2020will}, etc. 
Besides the COVID-19 pandemic, it is relevant to build dynamic models to study viral spreading processes to predict future outbreaks and to design control algorithms to mitigate the epidemic \cite{nowzari2016epidemics}. One of the popular ways to capture 
viral spreading processes is 
by
using 
network-based compartmental models\cite{mei2017epidemics_review}. In networked epidemic models, the infection rates, healing rates, and network structures all play important roles in determining the behaviors of the epidemic spreading processes. Recently, social factors such as human awareness\cite{Paarporn_tcss}, 
% traffic flow \cite{vrabac2020analysis}, 
opinion interactions\cite{weihao2020_opinion}, 
etc., are being
taken into consideration when modeling epidemic spreading over networks. In this work, we will consider the classical networked Susceptible-Infected-Recovered (SIR) model coupled with opinion dynamics.

In previous works, researchers have studied the networked SIR models from different perspectives. 
In \cite{mei2017epidemics_review}, the authors study
the dynamical behaviors of the networked SIR model, and 
analyze
the threshold conditions for an epidemic to increase or decrease. In \cite{hota2020closed}, %and \cite{zhang2021estimation} 
the authors leverage testing data to estimate the key parameters of the networked SIR model to design resource allocation methods to mitigate the epidemic. 

As mentioned before, 
people's beliefs in the seriousness of the epidemic 
is one important social factor that
will have an impact on the 
spreading process. 
For example, \cite{weitz2020awareness} studies the correlations between the awareness-driven behaviors during the COVID-19 pandemic and the spreading of the COVID-19. Further, \cite{bhowmick2020influence} constructs a multiplex network with a networked SEIV model coupled with opinion dynamics, then explores the disease-free equilibrium. Inspired by the health-belief model developed by social scientists \cite{healthbelief}, where people's behavior in the pandemic will be influenced by their beliefs in the seriousness of the epidemic, \cite{weihao2020_opinion} and \cite{she2021network} develop a networked SIS model with cooperative opinion dynamics, and both cooperative and antagonistic opinion dynamics, respectively. The authors in \cite{weihao2020_opinion} study both the disease-free and non-disease-free equilibria of the model. Our previous work, \cite{she2021network}, defines an opinion-dependent reproduction number to explore further the effect of the antagonistic opinions in epidemic spreading. Based on the health-belief model, we will develop a networked SIR model coupled with cooperative opinion dynamics.

The main contributions of this work can be summarized as follows: 
we define a networked SIR epidemic model coupled with opinion dynamics.
Then, we develop two concepts: 
an effective reproduction number and a peak infection time. %and an infection threshold. 
We utilize the 
effective reproduction number to explore epidemic spreading by studying the peak infection time. 
In particular, different from the previous works \cite{weihao2020_opinion,she2021network}, where stability and convergence of the equilibria are the main focuses, this work emphasizes more on exploring the transient behavior of the epidemic, characterized by the effective reproduction number and the peak infection time.
Additionally, we further analyze the opinion states via the behavior of the epidemic.
%\vspace{-2ex}
%\subsection{Paper Outline}

We organize the paper as follows: In Section II, we introduce the networked SIR model coupled with opinion dynamics and formulate the problems of interest; Section III studies the equilibrium of the developed model. Based on the model, Section III defines the effective reproduction number and peak infection time.
Section III further explores the epidemic's dynamical behavior by relating the effective reproduction number and the peak infection time.
Section IV illustrates the results of the paper through simulations. 
Section V concludes the paper and outlines research directions.
%\subsection{Background}
%\subsection{Outline}
\vspace{-5ex}
\subsection*{Notation}
For any positive integer $n$, we use $[n]$ to denote the index set
$\left\{ 1,2,\ldots,n\right\} $. We view vectors as column vectors
and write $x^{\top}$ to denote the transpose of a column vector $x$. %We use $x^{n}$ to denote the vector, of the same size as $x$, whose each entry equals the $n$th power of the corresponding entry of $x$. 
For a vector $x$, we use $x_{i}$ to denote the $i$th entry.
For any matrix $M\in\R^{n\times n}$, we use $[M]_{i,:}$, $[M]_{:,j}$,
$[M]_{ij}$, to denote its $i$th row, $j$th column and $ij$th
entry, respectively. We use $\tilde{M}=\diag\left\{ m_{1},\ldots,m_{n}\right\} $
to represent a diagonal matrix $\tilde{M}\in\R^{n\times n}$ with
$[\tilde{M}]_{ii}=m_{i}$, $\forall i\in\left[n\right]$. We use $\boldsymbol{0}$
and $\boldsymbol{1}_n$ to denote the vectors whose entries all equal 0 and 1,
respectively, and $I$ to denote the identity matrix. The dimensions
of the vectors and matrices are to be understood from the context.
%Let $\partial\left[c,d\right]^{n}$ and $\mathrm{Int}\left[c,d\right]^{n}$ denote the boundary and interior of the cube $\left[c,d\right]^{n},$ $c,d\in R$, respectively.

For a real square matrix $M$, we use $\rho\left(M\right)$ and $\sigma\left(M\right)$ 
%\phil{[\hmph{$s$ might be confused with susceptible} ... I agree, I just got confused in the discussion after Theorem~\ref{thm:1}]} 
to denote
its spectral radius and spectral abscissa (the largest real part among its eigenvalues), respectively. For any two vectors $v,w\in\R^{n}$,
we write $v\geq w$ if $v_{i}\geq w_{i}$, %$v>w$ if $v_{i}>w_{i}$, 
and $v\gg w$ if $v_{i}>w_{i}$, $\forall i\in\left[n\right]$.
The comparison notations between vectors are used for matrices as well, 
for instance, for $A,B\in\R^{n\times n}$, $A\gg B$ indicates
that $A_{ij}>B_{ij}$, $\forall i,j\in\left[n\right]$. 
%For any two sets $A$ and $B$, we use $A\setminus B$ to denote the set of
%elements in $A$ but not in $B$.
%We also employ a modified signum function:
%\[
%sgnm\left(x\right)=\begin{cases}
%\begin{array}{cc}
%1, & \mathrm{if}\,\,\,\text{\,\,\ensuremath{x\geq0}}\\
%-1, & \mathrm{if}\,\,\,\,\,x<0.
%\end{array} & \end{cases}
%\]
%The Dirac delta function, which is
%the first derivative of $\frac{1}{2}sgnm\left(\cdot\right)$, is represented by $\theta\left(\cdot\right)$.
% 
% \baike {
% [I moved the introduction of Graph theory here to simplify the Problem Formulation of both Epidemic Dynamics and Opinion Dynamics]

Consider a directed graph $\mathcal{G}=\left(\mathcal{V},\mathcal{E}\right)$, with the node set $\mathcal{V}=\left\{ v_{1},\ldots,v_{n}\right\} $
and the edge set $\mathcal{E}\subseteq \mathcal{V}\times \mathcal{V}$. Let
matrix $A \in\R^{n\times n}$, $[A]_{ij}=a_{ij}$, denote the adjacency matrix of $\mathcal{G}=\left(\mathcal{V},\mathcal{E}\right)$, where $a_{ij}\in\R_{>0}$
if $\left(v_{j},v_{i}\right)\in \mathcal{E}$ and $a_{ij}=0$ otherwise.
Graph $\mathcal{G}$ does not allow self-loops, i.e., $a_{ii}=0,$ $\forall i\in \left[n\right]$. Let
$k_{i}=\sum_{j\in \mathcal{N}_{i}}\left|a_{ij}\right|$, where $\mathcal{N}_{i}=\left\{ \left.v_{j}\right\vert \left(v_{j},v_{i}\right)\in \mathcal{E}\right\}$
denotes the neighbor set of $v_{i}$ and $\left|a_{ij}\right|$ denotes
the absolute value of $a_{ij}$. The graph Laplacian of $\mathcal{G}$ is defined
as $L\triangleq \tilde{K}-A$, where $\tilde{K}\triangleq \diag\left\{ k_{1},\ldots,k_{n}\right\}$.
%is a diagonal matrix.
% }
\vspace{-4ex}
\section{Modeling and Problem Formulation}
In this section, we 
% will 
introduce the networked SIR model coupled with opinion dynamics. We 
% will 
also
formulate the problem to be analyzed in this work.
%\phil{[it's not great to start a subsection right at the beginning of a section... it's better to at least give a preview of what will come in the section before jumping into the subsections]}
%\vspace{-2ex}
%\subsection{Basic SIR Model}

We start by defining a \textit{disease transmission network} $\mathcal{G}=\left(\mathcal{V},\mathcal{E}\right)$ as a weighted directed graph with a node set $\mathcal{V}=\left\{v_{1},\ldots,v_{n}\right\}$ representing $n$ disjoint communities and the edge set $\mathcal{E}\subseteq \mathcal{V}\times \mathcal{V}$ representing disease-transmitting contacts over $\mathcal{V}$. We denote the weight of each edge $(v_j, v_i)$ as $\beta_{ij}$.
Then, a basic continuous-time networked SIR model on graph $\mathcal{G}$, which was studied in \cite{mei2017epidemics_review}, can be defined as:
\vspace{-1ex}
\begin{subequations}
\begin{alignat}{3}
   \dot{s_{i}}\left(t\right) &= -s_{i}\left(t\right)\sum_{j\in\mathcal{N}_{i}}\beta_{ij}x_{j}\left(t\right), \label{eq:S}\\
    \dot{x_{i}}\left(t\right) &= s_{i}\left(t\right)\sum_{j\in\mathcal{N}_{i}}\beta_{ij}x_{j}\left(t\right) -\gamma_{i}x_{i}\left(t\right),\label{eq:I}\\
    \dot{r}_{i}\left(t\right) &= \gamma_{i}x_{i}\left(t\right)\label{eq:R},
\end{alignat}
\end{subequations}
where $(s_i(t), x_i(t), r_i(t)) \in [0, 1]$, $\forall i\in[n]$ are the states indicating the proportion of susceptible, infected, and recovered population in community $i \in 
% \{1, \hdots, n\} 
[n]$ at time $t\geq0$, respectively. Moreover, $\beta_{ij} \in \mathbb{R}_{\geq 0}$ is the \textit{transmission rate} from community $j$ to $i$, and $\gamma_{i} \in \mathbb{R}_{\geq 0}$ is the \textit{recovery rate} of community $i$. Note that \eqref{eq:S}-\eqref{eq:R} 
%\eqref{eq:S}, \eqref{eq:I}, \eqref{eq:R} 
satisfy $s(t) + x(t) + r(t) = 1\ \forall t\geq t_0 \in \mathbb{R}_{\geq0}$ as a result of the assumption that $\exists t_0 \in \mathbb{R}_{\geq0}$ such that $s(t_0) + x(t_0) + r(t_0) = 1$ and $\dot{s}(t) + \dot{x}(t) + \dot{r}(t) = 0\ \forall t \in \mathbb{R}_{\geq0}$.
 %Each community $i$ is assumed to be homogeneous sub-networks, meaning each of them is a fully connected network with an equal transmission rate $\beta_{ii}$ between its community members, and each member has the same recovery rate $\gamma_i$.
%\subsection{%Awareness
%Opinion Dynamics}

Similarly, we define 
the
\textit{opinion spreading network} as a directed graph $\bar{\mathcal{G}}=\left(\mathcal{V},\bar{\mathcal{E}}\right)$, where the edge set $\bar{\mathcal{E}}\subseteq \mathcal{V}\times \mathcal{V}$ represents the opinion-disseminating interactions over the same $n$ communities. Each edge in the graph is weighted by $\bar{a}_{ij} \in \mathbb{R}_{\geq 0}$ indicating the opinion-disseminating influence from node $j$ to node $i$. Let $o_i(t)\in[0,1]$, $\forall i\in[n]$, $t\geq 0$, denote the belief of community $i$ on the severity of the epidemic at time $t$, where $o_i(t)=1$ indicates community $i$ considers the epidemic to be extremely serious, while $o_i(t)=0$ implies
% communities believe 
community $i$ believes
the epidemic is not serious at all. We adapt 
% the 
Abelson’s models of opinion dynamics from 
% {Equation (10) in 
\cite[Equation (10)]{proskurnikov2017tutorial}, where $\bar{\mathcal{N}}_{i}=\left\{ \left.v_{j}\right\vert \left(v_{j},v_{i}\right)\in \bar{\mathcal{E}}\right\}$:
\vspace{-1ex}
\begin{equation} \label{eq:O}
\dot{o}_{i}\left(t\right)=\sum_{j\in \mathcal{\bar{N}}_{i}}\bar{a}_{ij}\left(o_{j}\left(t\right)-o_{i}\left(t\right)\right).
\end{equation}
%where $\mathcal{\bar{N}}_i \subseteq \mathcal{N} = \{1, \hdots, n\}$ is the indexed neighborhood of node $i$ in $\bar{G}$.
%\vspace{-4ex}
%\subsection{Epidemic-Opinion Model}

We assume that the $n$ communities share a homogeneous minimum incoming transmission rate $\beta_{\min}$ and a homogeneous recovery rate $\gamma_{\min}$, where $\beta_{\min}$ corresponds to the strongest belief of a community in the severity of the epidemic $o_i(t) = 1$, while $\gamma_{\min}$ corresponds to the weakest belief of a community in the severity of the epidemic $o_i(t) = 0$.
To couple the networked SIR model with the opinion dynamics, we employ the health-belief model, which is  the best known and most widely used theory in health behavior research~\cite{healthbelief}. The health-belief model proposes people's beliefs\footnote{ In this article, beliefs, attitudes, and opinions are used interchangeably.} about health problems, perceived benefits of actions, and/or perceived barriers to actions that can explain their engagement, or lack thereof, in health-promoting behavior. Therefore, people's beliefs in their perceived susceptibility and/or in their perceived severity of the illness affect how susceptible they are and/or how effective they will be at healing from these epidemics. We %are ready to 
define a networked SIR model influenced by the opinion dynamics as:
\begin{subequations}
\small
\begin{alignat}{3}
 \dot{s_{i}}\left(t\right) &= -s_{i}\left(t\right)\sum_{j\in\mathcal{N}_i}\left(\beta_{ij}-(\beta_{ij}-\beta_{\min})o_i(t)\right)x_{j}\left(t\right),  \label{eq:S-O}\\
\dot{x_{i}}\left(t\right) &=s_{i}\left(t\right)\sum_{j\in\mathcal{N}_i}\left(\beta_{ij}-(\beta_{ij}-\beta_{\min})o_i(t)\right)x_{j}\left(t\right)\nonumber\\
& \ \ \ \ -\left(\gamma_{\min}+(\gamma_i-\gamma_{\min})o_{i}(t)\right)x_{i}\left(t\right).\label{eq:I-O}
\end{alignat}
\end{subequations}
% \vspace{-1ex}
% \begin{align}
% \dot{s_{i}}\left(t\right) &= -s_{i}\left(t\right)\sum_{j\in\mathcal{N}_i}\left(\beta_{ij}-(\beta_{ij}-\beta_{\min})o_i(t)\right)x_{j}\left(t\right),  \label{eq:S-O}\\
% \dot{x_{i}}\left(t\right) &=s_{i}\left(t\right)\sum_{j\in\mathcal{N}_i}\left(\beta_{ij}-(\beta_{ij}-\beta_{\min})o_i(t)\right)x_{j}\left(t\right)\nonumber\\
% & \ \ \ \ -\left(\gamma_{\min}+(\gamma_i-\gamma_{\min})o_{i}(t)\right)x_{i}\left(t\right).\label{eq:I-O}
% \end{align}
%The above formulation is inspired by the health belief model introduced in \cite{}. 
In \eqref{eq:S-O} and \eqref{eq:I-O}, the transmission rate of community $i$, %similarly to the per-node recovery rate, in the basic SIR model 
is obtained through the linear interpolation between $\beta_{ij}$ and $\beta_{min}$ scaled by the level of community $i$'s belief in the seriousness of the epidemic, $o_{i}(t)$. 
% H
A higher $o_{i}(t)$ will lead to lower transmission rates for community $i$. A similar interpretation can apply to the healing rate of community $i$ which is scaled by the level of community $i$'s belief in the seriousness of the epidemic, with a higher opinion state 
% leads to
leading to a
higher healing rate of community $i$.

Notice that $(1-s_{i}(t))=x_{i}(t)+r_{i}(t)$, $t\geq0$, $\forall i\in[n]$, captures the proportion of the population that are infected/have been infected with the epidemic. Hence, $(1-s_{i}(t))$ captures the %severity of the disease spreading 
\textit{infection level}
within community~$i$. By modifying the opinion dynamics in \eqref{eq:O} via cooperating the infection level:% into the awareness dynamics:
%\begin{equation*}
%\label{eq:O-I}
%\dot{o_{i}}\left(t\right) =%\bar{a}_0
%\left(1-s_{i}\left(t\right)-o_{i}\left(t\right)\right)+\sum_{j\in\mathcal{\bar{N}}_i}\bar{a}_{ij}\left(o_{j}\left(t\right)-o_{i}\left(t\right)\right).
%\end{equation*}
%To incorporate the current epidemic severity into the awareness dynamics. 
%We further assume the homogeneity of the evolution of awareness caused by the community's infection severity $\left(1-s_{i}\left(t\right)-o_{i}\left(t\right)\right)$ and the outside communities' perception $\left(o_{j}\left(t\right)-o_{i}\left(t\right)\right)$, i.e. $\bar{a}_0 = \bar{a}_{ij} = 1\ \forall i,j \in \{1, \hdots, n\}$. This simplify the awareness dynamics:
\vspace{-1ex}
\begin{equation}
\dot{o_{i}}\left(t\right)  =\left(1-s_{i}\left(t\right)-o_{i}\left(t\right)\right)+\sum_{j\in\bar{\mathcal{N}}_i}\left(o_{j}\left(t\right)-o_{i}\left(t\right)\right), \label{eq:O-I}
\end{equation}
where a higher proportion of the infected plus recovered population within community $i$ will lead to 
a stronger belief in the seriousness of the epidemic, and vice versa.
%\vspace{-2ex}
%\subsection{Problem Formulation}

We have presented the epidemic-opinion model in \eqref{eq:S-O}-\eqref{eq:O-I}, then, we can state the problem of interest in this work. We are interested in exploring the mutual influence between the epidemic 
% \phil{dynamics} 
spreading over 
the graph $\mathcal{G}$ of
$n$ communities 
% captured by graph $\mathcal{G}$ 
in \eqref{eq:S-O} and \eqref{eq:I-O}, and the opinions of the $n$ communities about the epidemic captured by graph $\bar{\mathcal{G}}$ in \eqref{eq:O-I}.
In this paper, we will:
\begin{enumerate}
    \item analyze the equilibria of the system in \eqref{eq:S-O}-\eqref{eq:O-I}. %under different settings, 
    In particular, we connect the opinion states at the equilibrium to the infection level of the communities;
    \item define an %\textit{opinion-dependent reproduction number}
     effective reproduction number
    to characterize the spreading of the disease. In particular, we
    explore the transient behavior of the epidemic-opinion model by leveraging peak infection time; %and infection thresholds;
    \item illustrate the results through simulations.
\end{enumerate}
The analysis presented in this work can provide insights for decision-makers who aim to analyze disease spreading 
% through 
and its coupling with
the public's opinion towards the epidemic.
\label{section2}
\vspace{-5ex}
\section{Main Results}
\label{section3}
This section examines the mutual influence between the epidemic dynamics in \eqref{eq:S-O} and \eqref{eq:I-O}, and the opinion dynamics in \eqref{eq:O-I}. 
Particularly,
we construct the compact form of the incorporated system to define an effective reproduction number to explore the peak infection time of the %epidemic-opinion 
model. We also %define infection thresholds and  
analyze the evolution of the epidemic %\phil{the epidemics [`the epidemic' or `epidemics' but the first is probably better]} 
by using the effective reproduction number and peak infection time. %and infection thresholds.

We write \eqref{eq:S-O}, \eqref{eq:I-O}, and \eqref{eq:O-I} in a compact
form as follows:
\vspace{-1ex}
\begin{subequations}
\small
\begin{alignat}{3}
\label{eq:S-O-C}
\dot{s}\left(t\right)&=-\left(\tilde{S}\left(t\right)\left(B-\tilde{O}\left(t\right)\left(B-B_{\min}\right)\right)\right)x\left(t\right), 
% \end{equation}
\\
% \begin{align}
\dot{x}\left(t\right) & =\left(\tilde{S}\left(t\right)\left(B-\tilde{O}\left(t\right)\left(B-B_{\min}\right)\right)\right)x\left(t\right)\nonumber\\
 & \ \ \ \ \ -\left(G_{\min}+\left(G-G_{\min}\right)\tilde{O}\left(t\right)\right)x\left(t\right),\label{eq:I-O-C}
% \end{align}
\\
% \begin{equation}
\label{eq:O-I-C}
\dot{o}\left(t\right)&=\left(\boldsymbol{1}_{n}-s\left(t\right)\right)-\left(\bar{L}+I_{n}\right)o\left(t\right),
\end{alignat}
\end{subequations}
% \vspace{-1ex}
% \begin{align}
% % \begin{equation}
% \label{eq:S-O-C}
% \dot{s}\left(t\right)&=-\left(\tilde{S}\left(t\right)\left(B-\tilde{O}\left(t\right)\left(B-B_{\min}\right)\right)\right)x\left(t\right), 
% % \end{equation}
% \\
% % \begin{align}
% \dot{x}\left(t\right) & =\left(\tilde{S}\left(t\right)\left(B-\tilde{O}\left(t\right)\left(B-B_{\min}\right)\right)\right)x\left(t\right)\nonumber\\
%  & \ \ \ \ \ -\left(G_{\min}+\left(G-G_{\min}\right)\tilde{O}\left(t\right)\right)x\left(t\right),\label{eq:I-O-C}
% % \end{align}
% \\
% % \begin{equation}
% \label{eq:O-I-C}
% \dot{o}\left(t\right)&=\left(\boldsymbol{1}_{n}-s\left(t\right)\right)-\left(\bar{L}+I_{n}\right)o\left(t\right),
% \end{align}
% \phil{[is there a reason you aren't using a single align environment for these equations?]}
% {\color{blue}[BS: I was thinking to label each equation separately. I am not sure if I can label them separately by putting them in one align.]} \hmph{Yes, I think you can: https://tex.stackexchange.com/questions/103891/multiple-labels-in-aligned-environment}
%\phil{[yep, it works; see above]}
where $\tilde{S}\left(t\right)=\diag\left(s\left(t\right)\right)$, $\tilde{O}\left(t\right)=\diag\left(o\left(t\right)\right)$, %$[B]_{ij}=\beta_{ij}$, 
$G_{\min}$ and $G$ are diagonal
matrices, with $\left[G_{\min}\right]_{ii}=\gamma_{\min}$,
and $[G]_{ii}=\gamma_{i}$, $\forall i \in[n]$. Note that $\bar{L}$ is the Laplacian matrix of the opinion spreading graph $\mathcal{\bar{G}}$.
%\phil{[should we use $[B]_{ij}$ and $[G]_{ii}$ to be consistent?]}. 
By defining $B\left(o(t)\right)=(B-\tilde{O}\left(t\right)\left(B-B_{\min}\right))$, $G\left(o(t)\right)=(G_{\min}+\left(G-G_{\min}\right)\tilde{O}\left(t\right))$,
\vspace{-1ex}
\begin{subequations}
\small
\begin{alignat}{2}
\label{S-O-Final}
\dot{s}\left(t\right) &=-\left(\tilde{S}\left(t\right)B\left(o(t)\right)\right)x\left(t\right),
% \end{equation}
% \begin{equation}
\\
\label{I-O-Final}
\dot{x}\left(t\right) & =\tilde{S}\left(t\right)B\left(o(t)\right)x\left(t\right)-G\left(o(t)\right)x\left(t\right).
\end{alignat}
\end{subequations}

% \begin{align}
% % \begin{equation}
% \label{S-O-Final}
% \dot{s}\left(t\right) &=-\left(\tilde{S}\left(t\right)B\left(o(t)\right)\right)x\left(t\right),
% % \end{equation}
% % \begin{equation}
% \\
% \label{I-O-Final}
% \dot{x}\left(t\right) & =\tilde{S}\left(t\right)B\left(o(t)\right)x\left(t\right)-G\left(o(t)\right)x\left(t\right).
% \end{equation}
% \end{align}
For the epidemic spreading process, we assume that community $i$ can pass the virus to community $j$ through at least one directed path in the network $\mathcal{G}$, $\forall i, j \in[n]$, $i\neq j$. For the opinion spreading process, we assume that community $i$ can affect community $j$'s opinion through at least one directed path in $\mathcal{\bar{G}}$. Therefore, we have the following assumption for the epidemic and opinion spreading over the communities:

\begin{assumption}
\label{A1}
Suppose $\forall i\in[n]$, 
% there hold 
$s_{i}(0), x_{i}(0),  o_{i}(0)\in[0,1]$, $s_i(0)+x_i(0)+r_i(0)=1$,
$\gamma_{i}\ge\gamma_{{\rm min}}>0$, and $\beta_{ij}\ge\beta_{{\rm min}}>0$, $\forall j\in \mathcal{N}_{i}$. Further, both $\mathcal{G}$ and $\bar{\mathcal{G}}$ are strongly connected.
\end{assumption}
\vspace{-3ex}
\subsection{Equilibrium}
First we show the  model in \eqref{eq:S-O-C}-\eqref{eq:O-I-C} is well-defined. 

%{\color{blue}[add that $s(t)\gg0 or x(t)\gg0].}

\begin{lemma}
\label{lem: well-defined}
For the epidemic-opinion model in \eqref{eq:S-O-C}-\eqref{eq:O-I-C}, if $(s_{i}(0), x_{i}(0), o_{i}(0))\in[0,1]$, and $s_i(0)+x_i(0)+r_i(0)=1$, then $(s_{i}(t), x_{i}(t), o_{i}(t))\in[0,1]$, $\forall t>0$, $\forall i\in[n]$. 
\end{lemma}

The system defined in \eqref{eq:S-O-C}-\eqref{eq:O-I-C} is a group of polynomial ODEs defined over the compact set $[0,1]^{3n}$. Hence, the system %defined in
%in \eqref{eq:S-O-C}-\eqref{eq:O-I-C} 
is Lipschitz over  $[0,1]^{3n}$. 
To prove Lemma \ref{lem: well-defined}, we can show that all gradient vector fields on the boundary of the set $[0,1]^{3n}$ are either pointing towards the set's interior or tangential to the boundary~\cite{blanchini2008set}. 
The proof is similar to the proof of {\cite [Lemma 7]{she2021network}} and thus omitted here.

\begin{lemma}
\label{lem:decreasing}
If $s_{i}(0), x_{i}(0),  o_{i}(0)\in[0,1]$, $\forall i\in [n]$, the susceptible states, $s_{i}(t)$, are
monotonically decreasing.% while the recovered states, %$r_{i}\left(t\right)$, are monotonically increasing $\forall t \geq 0$.
\end{lemma}

\begin{proof}
Lemma \ref{lem: well-defined} indicates that $x(t)\in [0,1]^n$  $\forall t\geq0$. From Assumption \ref{A1}, the matrix $-(\tilde{S}(t)(B-\tilde{O}(t)(B-B_{\min})))$ is non-positive, therefore, the RHS of \eqref{eq:S-O-C} is always non-positive, which completes the proof.
% We prove that the susceptible states are monotonically decreasing.
\end{proof}

After considering the monotonicity of the susceptible population, we move to the next lemmas to study the equilibria of the epidemic-opinion model. %Lemma \ref{lem:non-zero} states that, no community can maintain a zero infection level during the epidemic spreading process. 

\begin{lemma}
\label{lem:Equi}
The equilibria of the epidemic-opinion model in \eqref{eq:S-O-C}-\eqref{eq:O-I-C} take the form $(s_{e},\boldsymbol{0},(\bar{L}+I_{n})^{-1}(\boldsymbol{1}_{n}-s_{e}))$, where $[s_e]_i\in[0,1]$ and $[(\bar{L}+I_{n})^{-1}(\boldsymbol{1}_{n}-s_{e}))]_i\in [0,1]$, $\forall i\in [n]$.
\end{lemma}
\begin{proof}
Based on Lemma \ref{lem:decreasing}, the RHS of \eqref{eq:S-O-C} is always less than or equal to zero. Therefore, the susceptible states are monotonically decreasing. 
To compute the equilibria of the system in \eqref{eq:S-O-C}-\eqref{eq:O-I-C}, let $\dot{s}(t)=\dot{x}(t)=\dot{o}(t)=\boldsymbol{0}$, which leads to $x_e=\boldsymbol{0}$, and $s_e$ can be any point located in the set $[0,1]^n$. Substituting the equilibrium point $s_e$ into \eqref{eq:O-I-C} with $\dot{o}(t)=\boldsymbol{0}$, %we have
\begin{equation}
\label{eq:equi_opi}
\boldsymbol{0}=(\boldsymbol{1}_{n}-s_e)-\left(\bar{L}+I_{n}\right)o_e.
\end{equation}
Since the row sums of the Laplacian matrix $L$ are zeros, the matrix $(\bar{L}+I_n)$ is a strictly diagonally dominant matrix, i.e., the matrix $(\bar{L}+I_n)$ is positive definite. Therefore,
\begin{equation}
\label{eq:unique}
o_{e}=\left(\bar{L}+I_{n}\right)^{-1}\left(\boldsymbol{1}_{n}-s_{e}\right).
\end{equation}
Now we have to show that for each $s_e\in [0,1]^n$, the solution $o_e\in[0,1]^n$. Rearrange 
% equation 
\eqref{eq:equi_opi} as follows:
\[
\left(\boldsymbol{1}_{n}-s_e\right)=\left(\bar{L}+I_{n}\right)o_e.
\]
If $s_e\in [0,1]^n$, we have $(\boldsymbol{1}_n-s_e)\in [0,1]^n$, which leads to $\left(\bar{L}+I_{n}\right)o_e\in [0,1]^n$. 

We prove $[o_e]_i\in[0,1]$ $\forall i\in [n]$ by contradiction. Without loss of generality, suppose that $[o_e]_i>1$, with $[o_e]_i-a_i=1$, belongs to a solution to the $i$th row of the equation 
$\left(\boldsymbol{1}_{n}-s_e\right)=\left(\bar{L}+I_{n}\right)o_e$.
As shown before, the matrix $\left(\bar{L}+I_{n}\right)$ is strictly diagonally dominant and $\bar{L}_{ii}+\sum_{j\in\mathcal{N}_{i}}\bar{L}_{ij}+1=1$, $\forall i\in[n]$. If the rest of the entries of the solution $[o_e]_j\leq [o_e]_i$, $j\neq i$, $j\in [n]$, we must have $[\left(I_{n}\right)o_e]_i>1$, which violates $(\boldsymbol{1}_n-s_e)\in[0,1]^n$. Hence, to ensure %the %\phil{%solution 
the $i$th entry
% the equation 
$[\left(\bar{L}+I_{n}\right)o_e]_i\in[0,1]$, %$\forall i\in [n]$,} %[what does this mean? $\left(\bar{L}+I_{n}\right)o_e$ is a vector and therefore a row sum doesn't make sense and further it's not an equation...]} 
%{\color{blue} to belong} %\phil{s [solution belongs]} 
%to $[0,1]$, 
there must be an entry $[o_e]_j>[o_e]_i$ in the solution $o_e$ such that $[o_e]_j-a_j=1$, where $a_j>a_i$, $j\in [n], j\neq i$. %\phil{[$j \neq i$?]}. 
%\phil{This [this what?]} is determined by
%{\color{blue}{since}} . If the rest of the entry $[o_e]_j\leq [o_e]_i$, $j\neq i$, $j\in [n]$, increases from $1$ to $1+a_i$, the total row sum of the $i$th row will increase by $n\times a_i$. 

%Therefore, to ensure the solution of \phil{the equation of the $i$th row [same comment as above... probably just replace with \eqref{}]} belong\phil{s [solution belongs]} to $[0,1]$, there must be at least \phil{one entries [one entry]} $[o_e]_j$, such that $[o_e]_j$ satisfies $[o_e]_j-a_j=1, a_j>a_i$, $j\in [n], j\neq i$ \phil{[this last part of this sentence is practically identical to the one three above... the logic is a bit hard to follow and I can't tell if it's the because of the English and/or there are holes]}. 
Then, consider the $j$th row of the equation $\left(\boldsymbol{1}_{n}-s_e\right)=\left(\bar{L}+I_{n}\right)o_e$. To ensure the solution of the $j$th row belongs to $[0,1]$, based on the same analysis technique, there must exist %\phil{an lease [`at least' or delete]}
one entry $[o_e]_k$ such that $[o_e]_k>1$, with $[o_e]_k-a_k=1$, $a_k>a_j>a_i$. Following the same process to check the rest of the rows of the equation, we can conclude that, for the last equation left to check, there is no such entry left in $o_e$ satisfying the inequality condition, $a_l>\dots>a_k>a_j>a_i$, where $[o_e]_l$ corresponds to the last row left to be checked in the equation $(\boldsymbol{1}_{n}-s_e)=(\bar{L}+I_{n})o_e$, with
$[o_e]_l-a_l=1$. Therefore, each entry of $o_e$ as a solution to $\left(\boldsymbol{1}_{n}-s_e\right)=\left(\bar{L}+I_{n}\right)o_e$ is not greater than 1.
%that $[o_e]_n-a_n=1$, such that $a_m>a_n$. Therefore, none of the solution of the equation \eqref{eq:O-I-C} is greater than 1. 
% Based on the same analysis, 
Using the same analysis technique, it can be shown that
each entry of $o_e$ as a solution to $\left(\boldsymbol{1}_{n}-s_e\right)=\left(\bar{L}+I_{n}\right)o_e$ is not smaller than 0. Therefore, the solution $o_e$ to the equation $\left(\boldsymbol{1}_{n}-s_e\right)=\left(\bar{L}+I_{n}\right)o_e$ must be located in $[0,1]^n$, which proves the statement.
\end{proof}

Lemma \ref{lem:Equi} shows that there are
infinite equilibria for the epidemic-opinion model captured by \eqref{eq:S-O-C}-\eqref{eq:O-I-C}. %Further, the proof of Lemma \ref{lem:Equi} also shows there must be 
% an 
%a
%unique opinion equilibrium in $[0,1]^n$, given an
% pairing with the 
%epidemic equilibrium. 
In particular, \eqref{eq:unique} indicates that the opinion states of the communities at the equilibrium can be uniquely evaluated as a function of the
% rest 
steady-state
susceptible population in the communities. 
% disregarding 
%\phil{irregardless}
%their initial evaluations of the epidemic. %\phil{[but does $s_e$ depend on their initial opinion value?]}. 
The following lemma further characterizes the condition that the communities reach a consensus on their opinions, i.e., the opinion states are the same when the epidemic disappears. %{\color{blue}[add a simulation to illustrate this lemma; add a simulation to compare the influences of the susceptible population to determine the opinion formation after the epidemic disappears]}. 

%When analyzing the equilibria of the system, in particular, in studying the equilibria of the opinion states, we find that each equilibrium of the opinion dynamics is determined by the epidemic states. Therefore, 

\begin{lemma}
\label{consensus}
The communities will reach consensus on their opinions if and only if all the communities have the same proportion of infections, captured by the equilibria $\left(s_{e},\boldsymbol{0},\boldsymbol{1}_n-s_{e}\right)$, where $[s_e]_i = [s_e]_j\ \forall i\neq j$.
\end{lemma}

\begin{proof}
First, we show the necessary condition. If all communities reach consensus at the equilibrium point, giving that $o_e$ is a consensus state, where $[o_e]_i=[o_e]_j$, $\forall i,j \in [n]$, $i\neq j$, then $Lo_e=\boldsymbol{0}$. From \eqref{eq:equi_opi}, $s_e=\boldsymbol{1}_n-o_e$, which implies that $[s_e]_i=1-[o_e]_i$, and $[s_e]_i=[s_e]_j$, $\forall i,j \in [n]$, $i\neq j$. 

For the sufficient condition, if all communities have the same proportion of the infected population at the equilibria, we have  $1-[s_e]_i=1-[s_e]_j$, $\forall i,j \in [n]$, $i\neq j$. Based on \eqref{eq:equi_opi}, we have $\boldsymbol{1}_n-s_e=(\bar{L}+I_n)o_e$. %Consider the case that $\left(\bar{L}+I_{n}\right)o_e=\boldsymbol{1}_n$. 
Since $\left(\bar{L}+I_{n}\right)$ is a positive definite matrix, $o_e=\left(\bar{L}+I_{n}\right)^{-1}(\boldsymbol{1}_n-s_e)$ is the unique solution to the equation. Further, $L$ is positive semi-definite with only one zero eigenvalue, paired with the eigenvector $v$, where all the entries of $v$ are the same. Therefore, $\left(\bar{L}+I_{n}\right)v=v$, 
% leading to 
giving that
$v=o_e=(\boldsymbol{1}_n-s_e)$ is the unique solution to the equation $\left(\bar{L}+I_{n}\right)o_e=\boldsymbol{1}_n-s_e$, completing
% . Therefore, we complete 
the proof.
\end{proof}

%As shown in the proof of Lemma \ref{eq:equi_opi}, $L_o$ is positive semi-definite with only one zero eigenvalue, corresponding to the eigenvector $v$, where all entries of $v$ are the same. Therefore, $\left(\bar{L}+I_{n}\right)v=v$, leading to $o_e$, where $o_i=o_j$, $\forall i,j \in [n]$, $i\neq j$, is the unique solution to the equation $\left(\bar{L}+I_{n}\right)o_e=s_e$, with $s_i=s_j$, $\forall i,j \in [n]$, $i\neq j$. 
%Due to all entries of $1-s_e$ are the same, while $\bar{L}+I_n$ has the smallest eigenvalue equal to one paired with an eigenvector in the consensus space. Therefore, $1-s_e$ is the eigenvector located in the consensus space of the matrix $\bar{L}+I_n$. Therefore, $o_e$ must located on the consensus manifold of the matrix $\bar{L}+I_n$.

%Lemma \ref{consensus} studies the situations where the communities can reach agreement on the severity of the epidemic. Further, Lemma \ref{consensus}
%Combined with Equation \eqref{eq:equi_opi}, 
%the proof 
%also suggests that 
%the proportion of infected population of community $i$ is determined by the opinion states, $o_e=\boldsymbol{1}_n-s_e$,

Lemma \ref{lem:Equi} and Lemma \ref{consensus} summarize the equilibria of the epidemic-opinion model from \eqref{eq:S-O-C} to \eqref{eq:O-I-C}. In particular, the lemmas 
show that the communities' beliefs in the seriousness of the epidemic can reflect the %are independent of their initial evaluations on the 
infection level. %\phil{[again, does $s_e$ depend on their initial opinion value? Since the infection and healing rates are functions of the opinions I don't see how this couldn't be...]}. 
More importantly, the communities will reach consensus on the seriousness of the epidemic if and only if the epidemic caused the same 
% amount 
proportion of infected population in all communities. Under this situation, the belief on the seriousness of the epidemic is proportional to the proportion of the recovered population in all communities, characterized by $o_e=\boldsymbol{1}_n-s_{e}$.
\begin{remark}
% Note that in general, it is not easy for the opinions of the communities to reach consensus on the seriousness of the epidemic, since the condition from Lemma \ref{consensus} requires that all the communities have the same infection level. However, there is one special case where this condition is not unreasonable, 
% % that 
% when the epidemic is so severe 
% % such 
% that every community is fully infected. 
% Then the communities will reach agreement that the epidemic is severe ($o_e = \boldsymbol{1}_n$).
The communities can rarely reach a consensus of their opinions on the epidemic's severity since it will be implied by Lemma \ref{consensus} that every community has the same infection level, which is unusual. However, one exception is when every community is fully infected, then all communities will agree that the epidemic is utterly severe ($o_e = \boldsymbol{1}_n$).
\end{remark}
\vspace{-3ex}
%\begin{lemma}
%Further, for the susceptible states and epidemic states, $s_i(t)$ and $x_i(t)$, $\forall t>0$, to stay at 0, the gradient vector field on the boundary of the compact set $[0,1]^{3n}$ must be zero, which means that $s_i(t)=0$, $x_i(t)=0$, $\forall i\in[n]$, $i\neq j$. A similar analysis technique can be used to show that if $s_i(t)$ and $x_i(t)$, $\forall t>0$, stay at 1, there must be $s_i(t)=1$, $x_i(t)=1$, $\forall i\in[n]$, $i\neq j$. 
%\label{lem:non-zero}
%If $0<s_i(t)<1, 0<x_i(t)<1$, $i\in[n]$, $\boldsymbol{0}\ll s(t)\ll \boldsymbol{1}\left(t\right)$, $\boldsymbol{0}\ll x(t)\ll \boldsymbol{1}\left(t\right)$ 
% $\forall t \geq 0$.}
%\phil{ [should this be the following?]
%If, for some $\tau \geq 0$, $0<s_i(\tau)<1, 0<x_i(\tau)<1$, $i\in[n]$, $\boldsymbol{0}\ll s(\tau)\ll \boldsymbol{1}\left(\tau\right)$, $\boldsymbol{0}\ll x(\tau)\ll \boldsymbol{1}\left(\tau\right)$ $\forall t \geq \tau$.
%\end{lemma}
\vspace{-1ex}
\subsection{Effective Reproduction Number}

\vspace{-1ex}

% After studying the equilibria of the epidemic-opinion model, we define the effective reproduction number for the model to characterize the dynamical behavior of the system. First we introduce the following 
% % existing 
% lemmas.
The effective reproduction number of the model characterizes the dynamical behavior of the system. We introduce the following lemmas before formally defining the notion.

\begin{lemma}
\label{lem:irr_non}\cite[Thm. 2.7, and Lemma 2.4]{varga2009matrix_book} Suppose
that M is an irreducible nonnegative matrix. Then: %the following statements hold:
\begin{enumerate}
\item M has a simple positive real eigenvalue equal to its spectral radius, $\rho(M)$;
\item There is a unique (up to scalar multiple) left eigenvector $v\gg\boldsymbol{0}$ (right eigenvector $w\gg \boldsymbol{0}$) pairing with $\rho(M)$;
\item $\rho(M)$ increases when any entry of M increases;
\item If N is also an irreducible nonnegative matrix with the same size and $M\geq N$, then $\rho(M)\geq\rho(N)$.
\end{enumerate}
\end{lemma}

\begin{lemma}
\label{lem:irr_spe}\cite[Prop. 1]{bivirus}
Suppose that $\varLambda$ is a negative diagonal matrix in $\R^{n\times n}$
and $N$ is an irreducible nonnegative matrix in $\R^{n\times n}$.
Let $M=\varLambda+N$. Then, $\sigma(M)<0$ if and only if $\rho(-\varLambda^{-1}N)<1$,
$\sigma(M)=0$ if and only if $\rho(-\varLambda^{-1}N)=1$, and $\sigma(M)>0$
if and only if $\rho(-\varLambda^{-1}N)~>~1$. 
\end{lemma}

\begin{lemma}
\cite[Sec. 2.1 and Lemma 2.3]{varga2009matrix_book}
\label{lem:irr_M} Suppose that $M$ is an irreducible Metzler matrix. Then, $\sigma\left(M\right)$ is a simple eigenvalue of $M$ and
there exists a unique (up to scalar multiple) left eigenvector $x\gg\boldsymbol{0}$ (right eigenvector $y\gg\boldsymbol{0}$) such
that $x^\top M=\sigma\left(M\right)x$ ($My=\sigma\left(M\right)y$). %
%Let $z>0$ be a vector in $\R^{n}$. If
%$Mz<\lambda z$, then $\sigma(M)<\lambda$. If $Mz=\lambda z$, then $\sigma(M)=\lambda$.
%If $Mz>\lambda z$, then $\sigma(M)>\lambda$.
\end{lemma}

\begin{lemma}
\label{lem:Metzler} 
\cite[Sec. 1, Lemma 2]{cvetkovic2020stabilizing}
Suppose that $A, B\in\R^{n\times n}$ are Metzler matrices. Then, $\sigma\left(A\right)\leq \sigma\left(B\right)$ if $A\leq B$.
\end{lemma}

\begin{definition}{[Effective 
Reproduction Number $R_o(t)$] 
Let $R_o(t)=\rho(G^{-1}(o(t))\tilde{S}(t)B(o(t)))$, $\forall t\geq0$, denote the %Opinion-Dependent 
Effective 
Reproduction Number, where $G(o(t))$, $\tilde{S}(t)$, and $B(o(t))$ are defined in \eqref{eq:I-O-C}.
\label{def:Rt} }
\end{definition}
Note that the 
effective 
reproduction number $R_o(t)$ depends not only on the proportion of the susceptible population $s(t)$, but also
on the variation of the opinion states \bs{$o\left(t\right)$}. %Therefore, it is named as opinion-dependent.

\begin{proposition}
\label{prop:Spe_r}The %Opinion-Dependent 
Effective 
Reproduction Number $R_o(t)$
%satisfies 
\bs{has} the following %conditions
\bs{properties}:
\begin{enumerate}
\item If
\[
% $
G^{-1}(o(t_1))\tilde{S}(t_1)B(o(t_1))\geq G^{-1}(o(t_0))\tilde{S}(t_0)B(o(t_0)),
\]
% $
% leads to 
then
$R_o(t_1)\geq R_o(t_0)$;

%\[
%\rho(G^{-1}(o(t_1))\tilde{S}(t_1)B(o(t_1)))\geq %\rho(G^{-1}(o(t_0))\tilde{S}(t_0)B(o(t_0));
%\]
\item $R_o(t)$ is strictly monotonically decreasing with respect to $s(t)$, $\forall t\geq 0$; %$\tilde{S}\left(t\right)$;
\item If $o\left(t_{0}\right)\leq o\left(t_{1}\right)$, $\forall t_0<t_1$, then 
$R_o(t_1)\leq R_o(t_0)$.
%then\textcolor{black}{{}
%$R_{t_{1}}^{o}\leq R_{t_{0}}^{o}$;}
 %{\color{blue}  [move this paragraph after proposition 1]}}
%{\color{blue}{[add a better notation for R]}}
\end{enumerate}
\end{proposition}

\begin{proof}
1) Based on Assumption \ref{A1}, and the definitions of 
%matrix 
$G(o(t))$ and $\tilde{S}(t)$, %[doesn't this follow from Lemma \ref{lem:non-zero}? {\color{blue}[BS: both G and S can be zero]}]} 
we conclude that $G(o(t))$ and $\tilde{S}(t)$
are 
positive definite diagonal matrices,
and $B(o(t))$
is an irreducible nonnegative matrix, $\forall t\geq0$. Hence, $G^{-1}(o(t))\tilde{S}(t)B(o(t))$
is an irreducible nonnegative matrix. For statement 1), %consider all of the opinions are the same at time $t_0$ and $t_1$, $t_0<t_1$, 
if 
\[
G^{-1}(o(t_1))\tilde{S}(t_1)B(o(t_1))\geq G^{-1}(o(t_0))\tilde{S}(t_0)B(o(t_0)),
\]
based on Lemma \ref{lem:irr_non},
\[
\rho(G^{-1}(o(t_1))\tilde{S}(t_1)B(o(t_1)))\geq \rho(G^{-1}(o(t_0))\tilde{S}(t_0)B(o(t_0))),
\]
which leads to $R_o(t_1)\geq R_o(t_0)$.

%\phil{[this sentence is really hard to read]}
2) $R_o(t)$ is strictly monotonically decreasing with respect to $s(t)$ means that, when $o(t)$ is fixed, 
% the decreasement of 
a decrease in
$s(t)$ leads to 
% the decreasement of  
a decrease in
$R_o(t)$, $\forall t\geq0$. Without loss of generality, assume that $o(t_0)=o(t_1)$, and $t_0< t_1$. 
%assuming that all opinions are the same at time $t_0$ and $t_1$ \phil{(), with [I think we need to explain this assumption a little bit and why don't we lose generality with it, i.e. point out that `with respect to $s(t)$' means that we fix the opinions]} , 
From Lemma \ref{lem:decreasing}, the proportion of infected population for each community is monotonically decreasing. Thus, $s(t_1)\leq s(t_0)$, and $[\tilde{S}(t_1)]_{ii}\leq [\tilde{S}(t_0)]_{ii}$, $\forall i\in[n]$, %Therefore, we %have
%$\tilde{S}^{-1}(t_1)<\tilde{S_{ii}}^{-1}(t_0)$, 
which leads to 

\vspace{-2ex}

\small
\[
\left[G\left(o\left(t_{0}\right)\right)^{-1}\tilde{S}(t_0)B\left(o\left(t_{0}\right)\right)\right]_{i,:}\geq\left[G\left(o\left(t_{1}\right)\right)^{-1}\tilde{S}(t_1)B\left(o\left(t_{1}\right)\right)\right]_{i,:}
\]
\vspace{-1.5ex}

\normalsize

\noindent
where other entries of both matrices remain the same.
Based on statement 1) of this proposition, $R_o(t_1)\leq R_o(t_0)$, implying that $R_o(t)$ is monotonically decreasing with respect to $s(t)$.
%\[
%G^{-1}(o\left(t_1\right))\tilde{S}(o\left(t_1\right))B({o}\left(t_1\right))\leq %G^{-1}(o\left(t_0\right))\tilde{S}(o\left(t_0\right))B({o}\left(t_0\right)).
%\]

%From the first and fourth statements of Lemma  \ref{lem:irr_non}, we have 
%\[
%\rho (G^{-1}(o\left(t_1\right))\tilde{S}(o\left(t_1\right))B({o}\left(t_1\right)))\leq \rho (G^{-1}(o\left(t_0\right))\tilde{S}(o\left(t_0\right))B({o}\left(t_0\right))).
%\]
%$R_{t_{0}}^{o}> R_{t_{1}}^{o}$. The same method can verify the case that $o\left(t_{0}\right)\geq o\left(t_{1}\right)$, then $R_{t_{0}}^{o}\leq %R_{t_{1}}^{o}$.

%\phil{[this sentence is really hard to read]}
3) When $o\left(t_{0}\right)\leq o\left(t_{1}\right)$, $\forall t_0<t_1$, we have 
%$G_{ii}\left(o\left(t_{0}\right)\right)\leq G_{ii}\left(o\left(t_{1}\right)\right)$.
\[
[G^{-1}\left(o\left(t_{0}\right)\right)]_{ii}\geq [G^{-1}\left(o\left(t_{1}\right)\right)]_{ii},
\]
\[
[B\left(o\left(t_{0}\right)\right)]_{i,:}\geq [B\left(o\left(t_{1}\right)\right)]_{i,:},
\]
where other entries of the matrices $G^{-1}\left(o\left(t_{1}\right)\right)$
and $B\left(o\left(t_{1}\right)\right)$ 
are equal to $G^{-1}\left(o\left(t_{0}\right)\right)$
and $B\left(o\left(t_{1}\right)\right)$, respectively. Additionally, from Lemma \ref{lem:decreasing}, $\forall t_0<t_1$, we have $\tilde{S}_{ii}(t_1)\leq \tilde{S}_{ii}(t_0)$. Following
% ed by 
the same analysis technique from the proof of statement 2), we 
% show 
have
that $R_o(t_1)\leq R_{o}(t_0)$.
\end{proof}

The effective reproduction number is influenced by both the opinion states and the proportion of the susceptible population. In particular, when the opinions are fixed, the susceptible proportion will always ensure 
that
the effective reproduction number 
% to
decreases, %[I think I understand what you're saying but this seems at odds with the simulations... maybe add `when the opinions are fixed' or something like that]}
since the recovered population will not be infected again. The opinion states will also have an influence on the change of the effective reproduction number in both directions: higher opinion states will lead to a lower effective reproduction number, and vice versa. As we mentioned in Section II, communities with stronger beliefs in the seriousness of the epidemic will take actions to avoid infections, leading to a lower effective reproduction number, and vice versa. %[this statement is really long and very hard to follow and I think it may be missing afew words... please double check it and highlight your edits.]}
Further, when all communities think the epidemic is extremely serious, %all the time,
$o_{\max}(t)=\boldsymbol{1}_n, \forall t\geq 0$. When  all communities think the epidemic is not worth treating at all during the pandemic, $o_{\min}(t)=\boldsymbol{0}$, $\forall t\geq 0$. Under 
% the two  different 
the two extreme
situations, the effective reproduction number %$R_o(t)$ 
satisfies the following %proposition.
% corollary.
result.
%\phil{[should this be a corollary of Proposition 1?]}
\begin{corollary}
\label{Prop:bounds}
For all $t\geq 0$, the effective reproduction number $R_{o}(t)$ satisfies $R_{\min}(t)
\leq R_{o}(t)\leq R_{\max}(t)$, where 
\begin{align*}
    R_{\min}(t)=\rho\left(G^{-1}\left(o_{\max}\right)\tilde{S}\left(t\right)B\left(o_{\max}\right)\right),
\end{align*}

\vspace{-1.5ex}

\noindent
and
\vspace{-1.5ex}
\begin{align*}
   R_{\max}(t)=\rho\left(G^{-1}\left(o_{\min}\right)\tilde{S}(t)B\left(o_{\min}\right)\right).
\end{align*}
\end{corollary}

The proof of Corollary \ref{Prop:bounds} is similar to the proof of Proposition \ref{prop:Spe_r}, thus omitted here. Corollary \ref{Prop:bounds} indicates that, given any time $t$, if the proportion of the susceptible population of each community are the same, the effective reproduction number is determined by the opinion states, where stronger beliefs in the seriousness of the epidemic lead to a lower effective reproduction number, and vice versa. Compared to the classical %network 
SIR model\cite{van2002reproduction}, where the effective reproduction number is monotonically decreasing with respect to the decreasement of the proportion of the susceptible population, under the influence of the opinions, the %effective reproduction number 
$R_o(t)$ defined in this work may not monotonically decrease. Therefore, $R_o(t)$  can lead to 
% a more various 
more diverse
behaviors in 
the
epidemic spreading process. In order to analyze the dynamical behavior of the epidemic-opinion model, we define a %\phil{threshold condition 
concept called peak infection time
%[the time isn't a condition...]} 
to characterize the influence of the effective reproduction number $R_o(t)$ in determining the behavior of the epidemic.
% spreading.

\vspace{-1.5ex}

\subsection{Peak Infection Time}

To connect the effective reproduction number $R_o(t)$ to the behavior of the epidemic-opinion model, we denote 
% $\sigma(t)=\sigma(\tilde{S}(t)B(o(t))-G(o(t)))$
\begin{equation}\label{eq:sigma}
    \sigma(t)=\sigma(\tilde{S}(t)B(o(t))-G(o(t)))
\end{equation}
and $p(t)$ as the spectral abscissa of $(\tilde{S}(t)B(o(t))-G(o(t)))$ and the corresponding normalized left eigenvector $\forall t\geq 0$, respectively. From Assumption \ref{A1} and Lemma \ref{lem:irr_M}, $(\tilde{S}(t)B(o(t))-G(o(t)))$ is an irreducible Metzler matrix, %and $(\tilde{S}(t))B_o(t)-G_o(t))>0$,
thus %$\sigma(\tilde{S}(t))B_o(t)-G_o(t))$,
$\sigma(t)$, $\forall t\geq0$, is a positive real eigenvalue. Additionally, the normalized left eigenvector $p(t)$ satisfies $p(t)\gg \boldsymbol{0}$ and  $p^\top(t)\boldsymbol{1}_n=1$, $\forall t\geq0$. Then, we define a weighted average of the epidemic states, for a given $t_1\in[t_0,t_2]$, $p^\top(t_1)x(t)$ as a metric to reflect the trend of the epidemic over the time interval $[t_0, t_2]$. %for any time $t\in[t_0,t_2]$. %where \baike{$p(t_0)$} is the normalized left eigenvector corresponding to  \baike{$\sigma(\tilde{S}(t_0))B(o(t_0))-G(o(t_0))$} at time $t_0$.  
%From Assumption \ref{A1} and Lemma \ref{lem:irr_M}, $(\tilde{S}(t))B_o(t)-G_o(t))$ is an irreducible Metzler matrix, %and $(\tilde{S}(t))B_o(t)-G_o(t))>0$,
%thus %$\sigma(\tilde{S}(t))B_o(t)-G_o(t))$,
%$\sigma(t)$, $\forall t\geq0$, is a positive real eigenvalue.  
Based on the properties of $\sigma(t)$ and $p(t)$, we have $p^\top(t_1) x(t)\geq 0$, $\forall t \geq 0$ and $p^\top(t_1) x(t)= 0$ if and only if $x(t)=\boldsymbol{0}$. 
% Moreover, %$\forall t \geq 0$, 
% $p^\top(t_1) x(t)$ is a weighted average of the epidemic states. %[this isn't a sentence... what are you trying to say? Maybe delete `since'?]}. 
Therefore, $p^\top(t_1) x(t)$ reflects the overall trend of the epidemic spreading over the time interval $[t_0,t_2]$, and $p^\top(t_1)x(t)=0$ if and only if the epidemic has died out. %\baike{Let $p(t_0)$ and $p(t_1)$ denote the normalized left eigenvectors corresponding to $\sigma(t_0)=\sigma(\tilde{S}(t_0))B(o(t_0))-G(o(t_0)))$ and $\sigma(t_1)=\sigma(\tilde{S}(t_1))B(o(t_1))-G(o(t_1)))$ at time ($t_0$) and ($t_1$), respectively, and  $p(t_p)$ denote the normalized left eigenvector corresponding to $\sigma(t_p)=\sigma(\tilde{S}(t_p))B(o(t_p))-G(o(t_p)))$, where $t_p$ is a peak infection time defined as follows:}
\begin{definition}{\label{def:PT}} [Peak Infection Time $t_p$]
A
Peak Infection Time $t_p$ is defined as a turning point, %\baike{where $\frac{d}{dt}p^\top(t_p)x(t)=0$ at $t_p$. Further
where $p^\top(t_p) x(t)$ is increasing for all $t\in[t_0,t_p)$ and $p^\top(t_p)x(t)$ is decreasing for all $t\in(t_p,t_1]$, for 
% any 
sufficiently small time intervals $(t_p-t_0)>0$ and $(t_1-t_p)>0$.
%where $R_o(t_p)=1$, and $R_o(t_p-\tau)<1$, $R_o(t_p+\tau)>1$, for any sufficient small $\tau>0$.
\end{definition}
%, and $p(t)$ is the normalized left eigenvector corresponding to $s(\tilde{S}(t))B_o(t)-G_o(t))$.

\noindent 
The peak infection time describes 
% the peak point 
a point in time
where the weighted average  of the infected proportions $p^\top(t_p) x(t)$ %\baike{ $p^\top(t_0) x(t)$ and } 
over the communities reaches a local peak value over the time interval $[t_0, t_1]$. %\phil{Note that [maybe add `due to (some previous result)'] there may not be a unique peak infection time.} 
%{\color{blue}[I moved this statement after Theorem]}
%For the time interval 
% $\tau$, 
%$[t_0, t_1]$, %with $\tau>0$, 
%the weighted average captured by \baike{ $p^\top(t_p) x(t)$} is increasing before reaching $t_p$, %and the weighted average \baike{$p^\top(t_1) x(t)$} 
%and is decreasing after passing~$t_p$.

\begin{theorem}
\label{thm:1}
% For 
Given a peak infection time $t_p$, we have
% \phil{[do we need to have a $\tau$ included, that is, should the peak infection time be defined in terms of some $\tau$]{\color{blue}[Yes, since there can be a point, such that $R_o(t_{p})=1$ but $R_o(t_{p}-\tau)<1$, and $R_o(t_{p}+\tau)>1$. I added a note latter to clarify the definition.]}} 
$R_{o}(t_p)=1$, $R_{o}(t)>1$, for $t\in [t_0, t_p)$ and $R_{o}(t)<1$, for $t \in (t_p, t_1]$, 
% with
for $t_p-t_0>0$ and  $t_1-t_p>0$ 
% being 
sufficiently small. %If $R_{o}(t_0)>1$, \baike{ $p^\top(t_0) x(t)$} will increase; if $R_{o}(t_1)<1$, \baike{ $p^\top(t_1) x(t)$} will decrease; and if $R_{o}(t_p)=1$, \baike{ $p^\top(t_p) x(t)$} will remain the same.
\end{theorem}
\begin{proof}
%In order to show the theorem, we first define the left eigenvector and right eigenvector corresponding to $[-G(o(t))+\tilde{S}(o(t))B(o(t))]$ as $p(t)$ and $q(t)$. Based on the second statement of Lemma \ref{lem:irr_non}, we have $p(t)\gg0$ and  $q(t)\gg0$. 
%We assume that $p(t)\gg0$ and  $q(t)\gg0$ are normalized eigenvectors as well. Additionally, we use both $p_{max}$ and $q_{max}$ to represent the normalized left and right eigenvectors corresponding to $R_{max}$. 
%Assuming all opinions states are fixed at $1\boldsymbol{e}$, we have $R_t^o=R_{max}$. 
%Consider the case where $R_{o}(t_0)>1$. 
First we show that for a peak infection time $t_p$, we have $R_o(t_p)=1$. %Since $p^{\top}(t_p)x(t)$ is a linear combination 
% of a groups 
%of ODEs, 
%$p^{\top}(t_p)x(t)$ 
Since $\frac{d}{dt}(p^{\top}(t_p)x(t))$ is a continuous function over the time interval $[t_0, t_1]$, based on Definition~\ref{def:PT}, $p^\top(t_p) x(t)$ is increasing for all $t \in [t_0,t_p)$ and %$p^\top(t_p) x(t)$ 
decreasing for all $t \in (t_p,t_1]$. Therefore, by continuity,
% of the derivative,  
$\frac{d}{dt}(p^{\top}(t_p)x(t))=0$ at time $t_p$.
Using this fact and multiplying $p^\top(t_p)$ on both sides of \eqref{I-O-Final}, 
% and considering the derivative at time $t_p$, 
we have
\begin{align*}
  0 &=
  \frac{d}{dt}(p^{\top}(t_p)x(t))\Big|_{t=t_p} 
  \\%& =p^{\top}(t_p)(\tilde{S}\left(t\right)B\left(o(t)\right)x\left(t\right)-G\left(o(t)\right)x\left(t\right))\\
 & =(p^{\top}(t_p))(\tilde{S}\left(t\right)B\left(o(t_p)\right)-G\left(o(t_p)\right))x\left(t_p\right)\\
 & = \sigma(t_p)p^{\top}(t_p)x\left(t_p\right)
 ,
% \\
% &=0,
\end{align*}
where the third equality follows from the definition of $\sigma(t_p)$ in \eqref{eq:sigma}.
% , and the second equality is a continuous function over the time interval $[t_0, t_1]$. 
%  From Definition \ref{def:PT},
% $p^{\top}(t_p)x(t)$ remains the same at $t_p$. 
Recall that $p(t_p)x(t)>0$ unless $x(t)=x_e=\boldsymbol{0}$. Thus, for $\frac{d}{dt}(p^{\top}(t_p)x(t_p))=0$, we must have $\sigma(t_p)=0$. Therefore, from Definition \ref{def:Rt} and Lemma \ref{lem:irr_spe}, $R_o(t_p)=1$.

Now we consider the time interval $[t_0, t_p)$. Since $x(t)$ and $o(t)$ are continuous functions in~$t$, 
% we have 
$\sigma(t)$
% the spectral radius of the matrix $G^{-1}(o(t))\tilde{S}(t)B(o(t))$ 
is also continuous in~$t$. 
% Therefore, $R_o(t)$ is continuous in $t$, leading to  $\sigma(t)$ being a continuous function. 
Then, for a given time $\tau \in [t_0, t_p)$, since $[t_0, t_p)$ is a sufficiently small time interval, we have $p(\tau)\approx p(t_p)$, by continuity.
%which implies there will be a small perturbation between  the entries of $p(t_p)$ and $p(t_1)$ if
%$(t_p-t_0)$ is sufficiently small. Therefore, we have 
From Definition \ref{def:PT}, since $p^\top(t_p) x(t)$ is increasing for all $t\in[t_0,t_p)$, $\frac{d}{dt}(p^{\top}(t_p)x(t))>0$ for $t\in [t_0, t_p)$. Using this fact and
multiplying $p^\top(t_p)$ on both sides of \eqref{I-O-Final}, %we have
\begin{align*}
 0&<\frac{d}{dt}(p^{\top}(t_p)x(t))\Big|_{t=\tau}\\
  %& =p^{\top}(t_p)(\tilde{S}\left(t\right)B\left(o(t)\right)x\left(t\right)-G\left(o(t)\right)x\left(t\right))\\
&=(p^{\top}(t_p))(\tilde{S}\left(\tau\right)B\left(o(\tau)\right)-G\left(o(\tau)\right))x\left(\tau\right)\\
&\approx (p^{\top}(\tau))(\tilde{S}\left(\tau\right)B\left(o(\tau)\right)-G\left(o(\tau)\right))x\left(\tau\right)\\
&=\sigma(\tau)(p^{\top}(\tau))x(t_\tau).
  %&=\sigma(t)p^{\top}(t)x(t).
\end{align*}
%\begin{align*}
 % \frac{d}{dt}(p^{\top}(t_p)x(t))%& 
 %& \approx \baike{(p^{\top}(t_0))}(\tilde{S}\left(t\right)B\left(%o(t)\right)-G\left(o(t)\right))x\left(t\right).
%\end{align*}
%From Definition \ref{def:Rt}, we have $\frac{d}{dt}(p^{\top}(t_p)x(t))>0$, which leads to $\sigma(t)p^{\top}(t)x(t)>0$. 
Since $p^{\top}(t)x(\tau)>0$, for $x(\tau)>0$, we have that $\sigma(\tau)>0$. Therefore, by Lemma \ref{lem:irr_spe}, $R(\tau)>1$ for any time $\tau \in [t_0,t_p)$.
% , $\sigma(\tau)>0$. 
Following the same analysis techniques, we can show that 
% $\sigma(\tau)<0$, and 
$R(\tau)<1$, for all $\tau\in(t_p,t_1]$, given a sufficiently small time interval $(t_p-t_1)>0$, completing the proof.
\end{proof}

Note that $R_o(t)=1$ is a necessary condition for the peak infection time, thus the condition does not guarantee that $t$ is the peak infection time. %the constraints that $R_o(t_p-\tau)>1$ and $R_o(t_p+\tau)<1$ are necessary. 
From Proposition 1, $R_o(t)$ is not a monotonic function with respect to $t$. Consider the case that $R_o(t_1)=1$, if, for $\epsilon>0$, $R_o(t_1-\epsilon)<1$ and $R_o(t_1+\epsilon)>1$, the time $t_1$ is not the peak infection time. Additionally, from Lemma \cite[Sec. 2.1 and Lemma 2.3]{varga2009matrix_book}, $p(t_p)$ is unique for a peak infection time $t_p$. 
%If there exists another $\bar{p}(t_p)$ satisfying that  $\frac{d}{dt}(\bar{p}^{\top}(t_p)x(t))=0$ at $t_p$, $\bar{p}^{\top}(t_p)$ will be an eigenvector corresponding to $\sigma(t_p)=0$. However, since $R(t)=1$ leads to $\sigma(t)=0$, at $t_p$, based on Lemma \ref{lem:irr_M}, we must have $\bar{p}(t_p)=p(t_p)$.

%Based Proposition \ref{prop:Spe_r} and Theorem \ref{thm:1}, we can concluded that there 
% might 
%may
%exist more than one peak infection time during the epidemic spreading process.
% \phil{[did we get simulations to back this up? If so, maybe allude to them here]}

%which offers the idea of changing the susceptible states and opinion states to ensure the epidemic 
% to reach 
%\phil{reaches a}
%peak time quicker, thus to mitigate the epidemic. Note that the condition $R_o(t)=1$ cannot guarantee the peak infection time, since there exist\phil{s} a case that $R_o(t-\tau)<1$, $R_o(t+\tau)>1$, where the weighted average of the epidemic states will decrease first then increase again \phil{have we shown this? Does it always happen?]}. Theorem \ref{thm:1} also indicates there might exist more than one peak infection time, 
% due to 
%\phil{since}
%$R_o(t)$ is not a monotonic function of $t$.

For $\forall t\in [t_1,t_2]$, from Lemma \ref{lem:decreasing} and \eqref{I-O-Final}, we have
\begin{equation}
(\tilde{S}\left(t_1\right)B\left(o_{\min}\right)-G\left(o_{\min}\right))\geq (\tilde{S}\left(t\right)B\left(o(t)\right)-G\left(o(t)\right).  
\label{eq:ineq}
\end{equation}
Based on Corollary \ref{Prop:bounds}, $R_{\min}\leq R_o(t)\leq R_{\max}(t)$, $\forall t\geq0$, %Moreover, from Lemma \ref{lem:decreasing}, $R_{\max}(0)>R_{\max}(t), \forall t>0$.
we define $\sigma_{\max}(t)=\sigma(\tilde{S}(t)B(o_{\min})-G^{-1}(o_{\min}))$, corresponding to $R_{\max}(t)$. %corresponding to the normalized left eigenvector $p^\top(0)$. %
Since $(\tilde{S}(t)B(o(t))-G^{-1}(o(t))$ is a Metzler matrix $\forall t$, from Lemma \ref{lem:Metzler}, we have $\sigma_{\max}(t_1)\geq \sigma_{\max}(t)\geq \sigma(t)$, $\forall t\geq t_1$.
%consider the weighted average $p(0)$, 
Then, we define $p_{\max}(t_1)$ corresponding to $\sigma_{\max}(t_1)$, and  multiplying $p_{\max}(t_1)$, on both sides of \eqref{I-O-Final},
%Similar to the proof of Theorem~\ref{thm:1}, $\forall t\geq t_1$, we have 

\vspace{-2ex}

\small
\begin{align*}
  \frac{d}{dt}(p_{\max}^{\top}(t_1)x(t)) &=p_{\max}^{\top}(t_1)((\tilde{S}\left(t\right)B\left(o(t)\right)-G\left(o(t)\right))x\left(t\right)).
 %& =(p^{\top}(t_1)(\tilde{S}\left(t_1\right)B\left(o(t_1)\right)-G\left(o(t_1)\right))x\left(t_1\right)),\\
% & = \sigma(t_1)p^{\top}(t_1)x(t_1).
\end{align*}

\vspace{-1ex}

\normalsize

\noindent
%For $\forall t\in [t_1,t_2]$, from Lemma \ref{lem:decreasing} and \eqref{I-O-Final}, we have
%$(\tilde{S}\left(t_1\right)B\left(o_{\min}\right)-G\left(o_{\min}\right)\geq (\tilde{S}\left(t\right)B\left(o(t)\right)-G\left(o(t)\right)$,
Then, based on \eqref{eq:ineq},

\vspace{-2ex}

\small
\begin{align}
\frac{d}{dt}(p_{\max}^{\top}(t_1)x(t)) 
 &\leq p_{\max}^{\top}(t_1)((\tilde{S}\left(t_1\right)B\left(o_{\min}\right)-G\left(o_{\min}\right))x\left(t\right)) \nonumber \\
 & = \sigma_{\max}(t_1)p_{\max}^\top (t_1)x\left(t\right), \label{eq:cor}
\end{align}

\vspace{-1ex}

\normalsize

\noindent
which leads to
\[
p_{\max}^{\top}(t_1)x\left(t\right)\leq p_{\max}^{\top}(t_1)x(t_1)e^{(\sigma_{\max}(t_1))t},
\]
for any $t\geq t_1$.
%\begin{align*}
%\label{eq:cor}
%\frac{d}{dt}(p_{max}^\top (t_1)x(t))=\sigma(t)p_{max}^\top (t_1)x\left(t\right)=p^{\top}(t_p)(\tilde{S}\left(t\right)B\left(o(t)\right)x\left(t\right)-G\left(o(t)\right)x\left(t\right))\\
%& =\leq \sigma_{\max}(t_1)p_{max}^\top (t_1)x\left(t\right),\\
%& = s_{max}p^{T}(t)x\left(t\right)    
%\end{align*}
%we have
%\begin{align*}
    %\frac{d}{dt}(p^{T}(t)x(t)) &
    %\leq s_{max}p^{T}(t)x\left(t\right),\\
%\end{align*}
%which leads to 
%\[
%p^{\top}(t)x\left(t\right)\leq %p^{\top}(t_1)x(t_1)e^{(\sigma_{\max}(t_1))t}.
%\]
The inequality listed above indicates that, when $\sigma_{\max}(t_1)<0$, the weighted average $p_{\max}^{\top}(t_1)x(t)$ will decrease exponentially fast to zero, $\forall t\geq t_1$, implying that $x(t)$ will decrease exponentially fast to zero. From Lemma~\ref{lem:irr_spe} and 
% Proportion 
% Proposition
Corollary~\ref{Prop:bounds}, $\sigma_{\max}(t_1)<0$ leads to $R_{\max}(t_1)<1$, which guarantees $R_{o}(t)<1$, $\forall t\geq t_1$. Hence, we have the following corollary, where we define $t_1$ as $t_f$. %as the final peak infection time, denoted
% ing
%by $t_f$. 
%Since the term $(p^{T}(t)x(0))e^{s(t)t}$ exponentially decreases
%when $s(t)<0$, and  $(p^{T}(t)x(0))e^{s(t)t}$ %exponentially increases if $s(t)>0$.
%Further, based on Lemma $\ref{lem:irr_spe}$, $s_{t}<0$ leads to $R_{t}<0$, and $s_{t}>0$ leads to $R_{t}>0$, which demonstrate the theorem.
\begin{corollary} 
\label{cor:c1}
If $R_{\max}(t_f)<1$, there will exist no peak infection time in 
% $t\geq [>?] t_f$ 
$(t_f, \infty)$,
%\phil{\geq0}$ \phil{[{it may be more interesting to replace 0 with some $T\geq 0$ where there is no peak infection time $t_p\geq T$} ]}, 
and $p_{\max}^{\top}(t_f)x(t)$ $\forall t\geq t_f$  will monotonically decrease to zero exponentially fast, indicating that the epidemic will die out exponentially fast.
\end{corollary}
%\begin{proof}
%At the beginning stage, we have  $R_{o_{max}}(0)<1$ leading to $s_{o_{max}}(0)<0$. Therefore, we have 
% 
%\begin{align*}
%\frac{d}{dt}(p_{max}^{T}(0)x(t)) & =p_{max}^{T}(0)(\tilde{S}\left(t\right)B_{o}\left(t\right)x\left(t\right)-G_{o}\left(t\right)x\left(t\right))\\
 %& =p_{max}^{T}(0)(\tilde{S}\left(t\right)B_{o}\left(t\right)-G_{o}\left(t\right))x\left(t\right)\\
% & \leq p_{max}^{T}(0)(\tilde{S}\left(0\right)B_{o_{max}}(t)-G_{o_{max}}\left(t\right))x\left(t\right)\\
 %& =s_{o_{max}}(0)p^{T}(t)x\left(t\right),
%\end{align*}
% 
%which means that $\frac{d}{dt}(p^{T}(t)x(t))\leq(p^{T}(t)x(0))e^{s_{o_{max}}(0)t}$,
%for all $t\geq0$. Therefore, $\dot{x}(t)$ decays monotonically and
%exponentially to zero for all $t>\tau$.
%\end{proof}
\noindent
Corollary \ref{cor:c1} connects $R_o(t)$ to the behavior of the epidemic process. In particular, for $t_f=0$, at the beginning stages of the epidemic, even with every community ignoring the epidemic, we still have $R_{\max}(0)<1$ which means that the epidemic is serious, and will disappear quickly. 

% Besides Corollary \ref{cor:c1}, 
In addition to Corollary \ref{cor:c1}, Theorem \ref{thm:1} also implies that, if the effective reproduction number $R_{\min}(t)$ at the beginning stages of the epidemic is greater than $1$,  before disappearing, there must exist at least one peak infection time.
This phenomenon is captured by the following corollary. 

\begin{corollary}
\label{cor:c2}
If $R_{\min}(0)>1$, then
\begin{enumerate}
    \item there will be at least one peak infection time $t_p$;
    \item $p_{\min}^{\top}(t_0)x(t)$ will increase exponentially fast before reaching a peak infection time $t_p$; %\phil{[what about when $R_o(t)=1$? It doesn't do either, right? So is statement 2) correct?]}
    %\item$p^{T}(t)x(t)$ will exponentially decrease to zero after reaching the \phil{final [max?]} peak infection time;
\end{enumerate}
\end{corollary}
\begin{proof}
1) First we prove that when $R_{\min}(0)>1$, there will be at least one peak infection time $t_p$. %Note that, based on Lemma \ref{lem:non-zero}, we have $x(t)\gg\boldsymbol{0}$, $\forall t>0$. 
Since $\dot{s}(t)\leq \boldsymbol{0}$, and $s(t)$ is lower-bounded by $\boldsymbol{0}$, we must have %{\color{blue}a moment when $s(t)=\boldsymbol{0}$,} %time $t_f$, the final peak infection time such that $\sigma(t_f)=0$, %or $x(t_ wf)=\boldsymbol{0}$ [how is this consistent with Lemma \ref{lem:non-zero}?]}, 
an equilibrium at $t=t_1$ when $\dot{s}(t_1)=\boldsymbol{0}$. Consider the case $s(t_1)=\boldsymbol{0}$ leading to $\dot{s}(t_1)=\boldsymbol{0}$. If $x(t_1)\neq \boldsymbol{0}$, from \eqref{I-O-Final}, we have $\dot{x}(t_1)=-G_o(t_1)x(t_1)< \boldsymbol{0}$, which violates Lemma~\ref{lem:Equi} that $(\boldsymbol{0},\boldsymbol{0},\boldsymbol{1}_n)$ is an equilibrium of the system. Therefore, in order to ensure $x(t)$ converges to $x(t_1)=\boldsymbol{0}$, based on Definition~\ref{def:PT} and Theorem \ref{thm:1}, 
for $\epsilon>0$, there must exist a moment $t_p$ where $t_p+\epsilon<t_1$, such that $R_o(t_p+\epsilon)<1$. Consider another case that $s(t_1)\neq\boldsymbol{0}$: to ensure $\dot{s}(t_1)=\boldsymbol{0}$,  we must have $x(t_1)=\boldsymbol{0}$. Thus, for the same reason, there must exist a time $t_p+\epsilon<t_1$ such that $R_o(t_p+\epsilon)<1$. Additionally, Since $R_{\min}(0)>1$, we can conclude that for both cases, there exists a moment $t_p-\epsilon<t_p$ such that $R_o(t_p-\epsilon)>1$. Therefore, from Theorem \ref{thm:1}, we have proved statement~1).

2) Consider one peak infection time $t_p$, and a time interval $[t_0, t_p-\delta_t]$, for sufficiently small $\delta_t>0$. Based on Lemma~\ref{lem:Metzler}, we have
$\sigma_{\min}(t_0)=\sigma(\tilde{S}(t_0)B(o_{\max})-G^{-1}(o_{\max}))>0$, pairing with the normalized left eigenvector $p_{\min}^{\top}(t_0)$.  %and $p_{min}$ be the corresponding left eigenvector. %Consider the time interval where the effective reproduction number is greater than 1, 
%From the proof of Theorem \ref{thm:1}, 
%From~\eqref{eq:cor},
Multiplying $p_{\min}^\top(t_0)$ on both sides of \eqref{I-O-Final}, and evaluating the derivative at $t_0$,
\vspace{-1ex}
\begin{align*}
    \frac{d}{dt}(p_{\min}^{\top}(t_0)x(t))|_{t=t_0} &
    = \sigma_{\min}(t_0)p^{\top}(t_0)x\left(t_0\right),
\end{align*}
$\forall t \in [t_0, t_p-\delta_t]$.
%which leads to 
%\[
%p^{\top}(t_0)x(t_0)= %p^{\top}(t_0)x(t_0)e^{(\sigma_{\min}(t_0))t}
%\]
% Due to 
Since
$\sigma_{\min}(t_0)>0$, $p_{\min}^{\top}(t_0)x(t)$ grows exponentially fast 
% over
for $t\in[t_0, t_p-\delta_t]$, completing the proof. %The same approach as in the proof of Corollary~\ref{cor:c1} can be used to prove the case that  $p^{T}(t)x(t)$ decrease exponentially after passing $$.
%Further, if $R_o(t)>1$, we have $\sigma_o(t)>1$ $\forall t$, there must exist a peak time $t_f$ such that $p^{\top}(t)x(t)$ grows exponentially fast over $[t,t+\tau]$, proving the second statement.
%\phil{[what about when $R_o(t)=1$? It doesn't do either, right? So is statement 2) correct?]}
%fast{\color{blue}add more explanation that $R_t = 1$ may not be the peak infection point. Add more explanation}.
%For the third statement, based on Corollary \ref{cor:c1}, the epidemic will decrease to zero exponentially fast after passing the last peak infection time {\color{blue}[BS:reformulate the Corollary].} \phil{[refer to $t_f$ notation in this paragraph?]}
%leading to $\dot{x}(t)\leq \boldsymbol{0}$. If $s(t)$ converges to $s_e> \boldsymbol{0}$, then $x(t)$ still converges to $\boldsymbol{0}$. Therefore, for any $(s(0), x(0))$, the trajectory $(s(t), x(t))$ converges to some equilibria with the form $(s_e, x_e)$. Let $s(t)=s_e+\delta_s(t)$, $x(t)=x_e+\delta_x(t)$, with $\delta_s(t)\geq \boldsymbol{0}$ and $\delta_x(t)\geq \boldsymbol{0}$. Additionally, we have $\delta_s(t)$ monotonically non-increasing and converges to 0.
%Linearizing the epidemic dynamics in \eqref{S-O-Final} and \eqref{I-O-Final} around $(s_n, \boldsymbol{0})$, we have 
\end{proof}

%\begin{corollary}
%For the epidemic-opinion model in, if the epidemic state in \eqref{eq:I-O-C} converges to $\boldsymbol{0}$ exponentially fast, we have \eqref{eq:S-O-C} converges to $\boldsymbol{0}$ exponentially fast.
%\end{corollary}
%\begin{proof}
%If the epidemic state in \eqref{eq:I-O-C} converges to $\boldsymbol{0}$ exponentially fast, there must exist a time $\tau$ such that $R_o(\tau)<1$, leading to $(\tilde{S}\left(t\right)B_{o}\left(t\right)x\left(t\right)-G_{o}\left(t\right)x\left(t\right))<0$.
%\end{proof}

\vspace{-1ex}

% Combine
By combining Corollaries \ref{cor:c1} and \ref{cor:c2} with Theorem \ref{thm:1}, 
% corollaries \ref{cor:c1} and \ref{cor:c2}, 
we can connect 
% that
the
behavior of the system 
% from 
in
\eqref{eq:S-O-C}-\eqref{eq:O-I-C} to the peak infection time of the system in the following theorem.

\begin{theorem}
\label{thm:2}
For the epidemic-opinion model 
% considered from \eqref{eq:S-O-C} to \eqref{eq:O-I-C}, 
in
\eqref{eq:S-O-C}-\eqref{eq:O-I-C},
the system will converge to an equilibrium $(s_{e},\boldsymbol{0},(\bar{L}+I_{n})^{-1}(\boldsymbol{1}_{n}-s_{e}))$, and the convergence is exponentially fast. % During the convergence, if $R_{\max}(0)<1$, the epidemic states decrease
% decreases
%; if $R_{\min}(0)>1$, the epidemic states will pass at least one peak infection time before decreasing to zero.
\end{theorem}

\begin{proof}
% Based on Theorem \ref{thm:1}, corollaries \ref{cor:c1} and \ref{cor:c2}, 
By combining Corollaries \ref{cor:c1} and \ref{cor:c2} with Theorem \ref{thm:1}, we can conclude that for any sufficiently small time interval $[t_0,t_1]$, we can find a weighted average of the infected epidemic state changes exponentially fast. Additionally,  Corollary \ref{cor:c1} indicates that the epidemic state will converge to zero exponentially fast, after passing $t_f$. %states will  
Therefore, We can conclude that the epidemic state $x(t)$ always converges to $\boldsymbol{0}$ exponentially fast. From \eqref{eq:S-O-C}, under the condition that $x(t)$  converges to $\boldsymbol{0}$ exponentially fast, $s(t)$ will converge to an equilibrium point exponentially fast. Consider the convergence of the opinion states $o(t)$, in  \eqref{eq:O-I-C}. The linear system 
\vspace{-1.5ex}
\[
\dot{o}\left(t\right)=-\left(\bar{L}+I_{n}\right)o\left(t\right)
\]

\vspace{-1.5ex}

\noindent
converges to zero exponentially fast due to the fact that all the eigenvalues of the system matrix $-(\bar{L}+I_{n})$ are smaller than zero. Note that \eqref{eq:O-I-C} is input-to-state stable since the linear system mentioned above has a unique globally stable equilibrium at $o=\boldsymbol{0}$. Therefore, 
% with input $s(t)$ 
treating $s(t)$ as an input to \eqref{eq:I-O-C} that
converges to $s_e$ exponentially fast, we have that $o(t)$ converges to $o_e$ exponentially fast\cite{khalil2002nonlinear}. Thus, we have proved the theorem.
%\phil{For the epidemic states during the convergence, we can obtain the statement directly from Theorem \ref{thm:1} [this is super vague and unclear], C}orollaries \ref{cor:c1}, and \ref{cor:c2}.
\end{proof}
\noindent
Combined with Theorem \ref{thm:1}, Corollaries \ref{cor:c1} and \ref{cor:c2}, Theorem \ref{thm:2} implies that the epidemic will die out eventually, 
% studies the behavior of the system in \eqref{eq:S-O-C}-\eqref{eq:O-I-C}. In particular, 
but the effective reproduction number $R_o(t)$ will determine whether there will be an outbreak or the epidemic will die out directly.

\vspace{-2ex}
	
\section{Simulation}	

\vspace{-1ex}

In the section, we will illustrate the main results developed in this work via simulations. Consider the epidemic coupled
with opinions spreading over ten communities. %\phil{via the same network structure [this makes the third sentence redundant... remove either this clause or the beginning of the third sentence]}. 
The epidemic and opinion spreading network satisfies Assumption \ref{A1}, and share the same graph topology $\mathcal{G}$ as shown in Fig.~\ref{fig_graph}.
% Note that we use the same graph structure in $\mathcal{G}$ to capture the epidemic and opinion graphs to simplify the simulations, and our results still apply to communities with different epidemic and opinion interactions. 
\begin{figure}
\begin{centering}
\includegraphics[trim = 7.5cm 7.2cm 4.8cm 6.15cm, clip, width=\columnwidth]{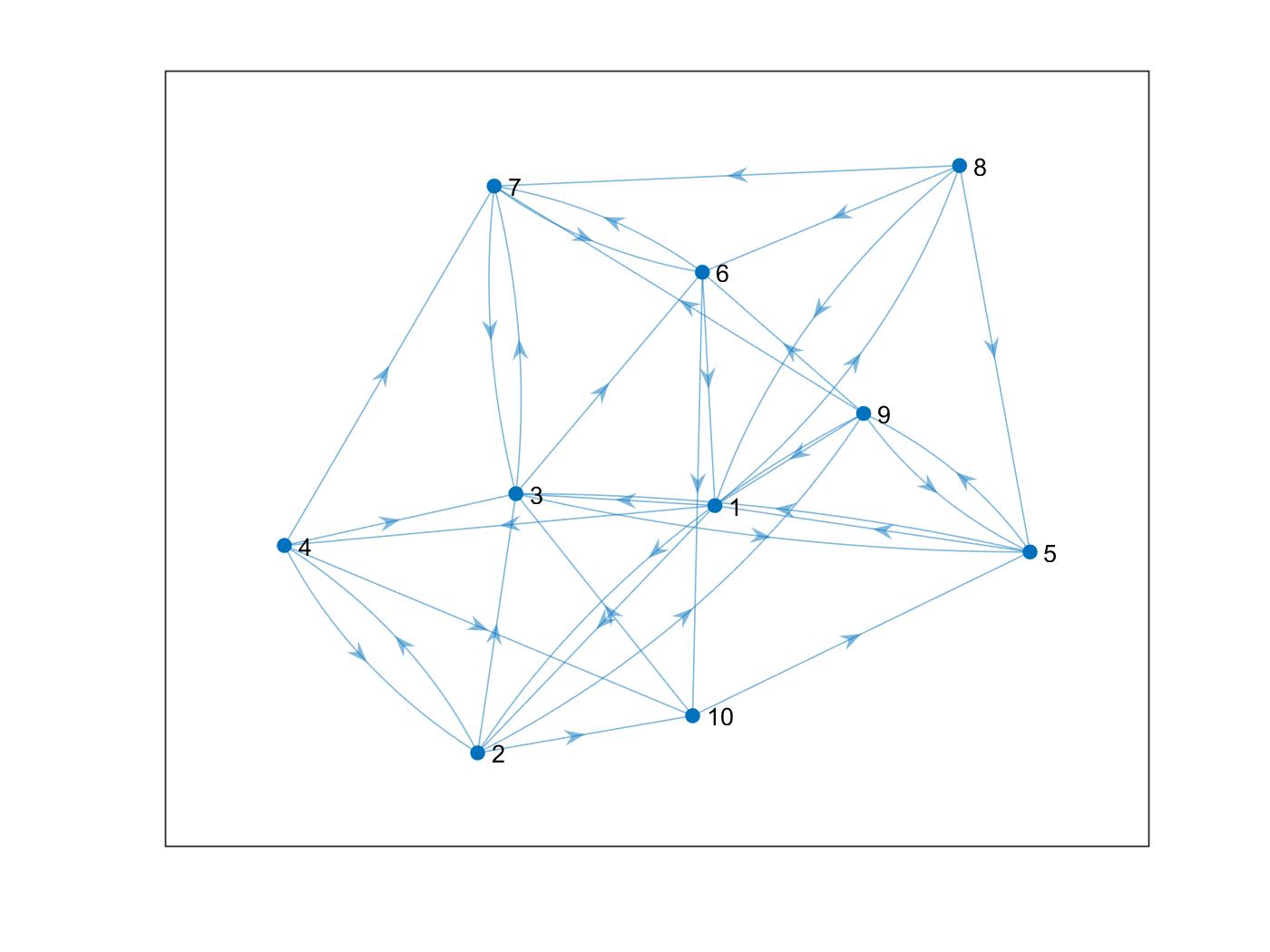}
\par\end{centering}
\centering{}\caption{The graph topology $\mathcal{G}$ of the simulated epidemic and opinion interactions
}
\label{fig_graph}
\end{figure}
\begin{figure}
\begin{centering}
\includegraphics[width=\columnwidth]{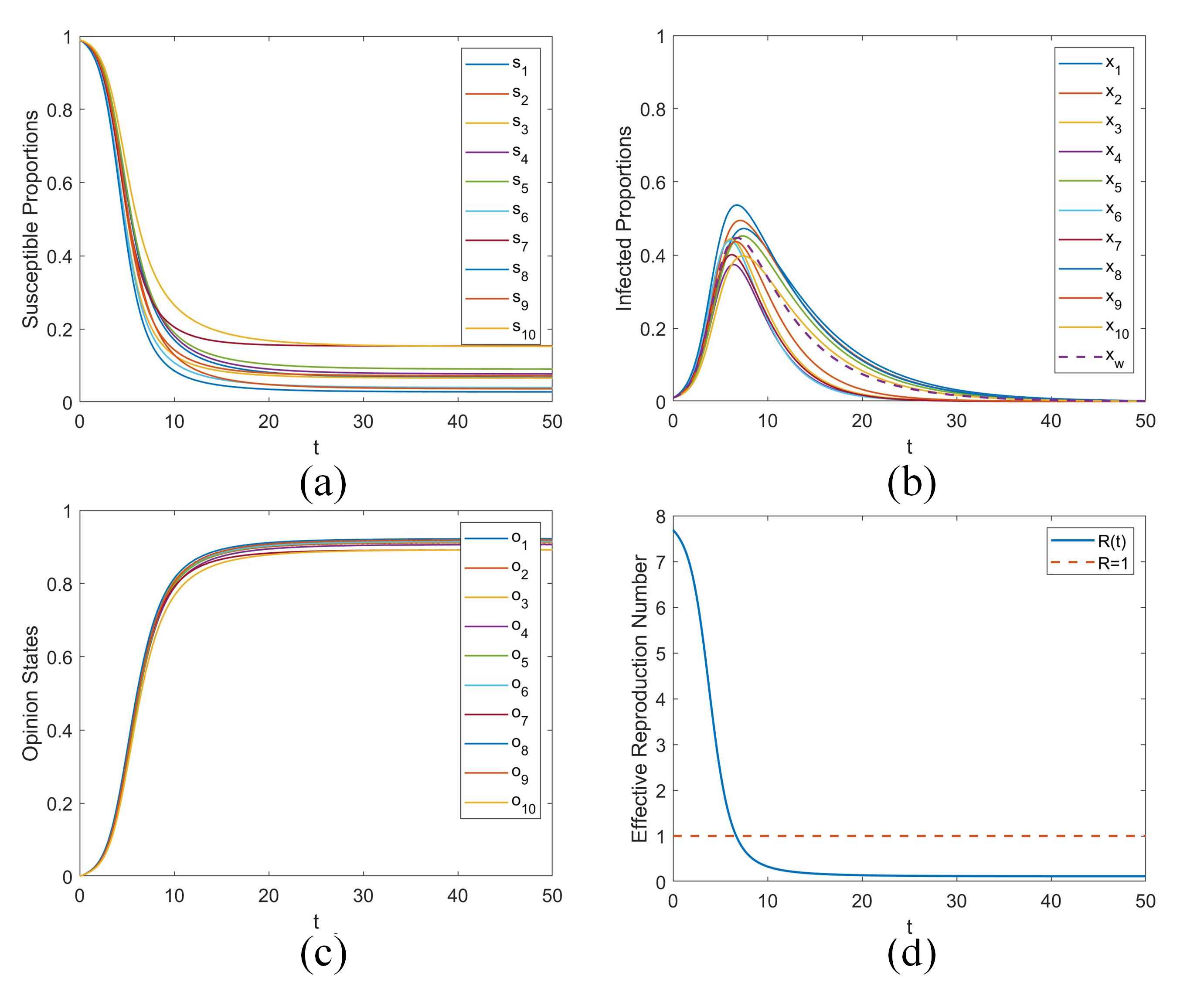}
\par\end{centering}
\centering{}\caption{Typical evolution of epidemics and opinions}
\label{simulation1}
\vspace{-2ex}
\end{figure}

% To observe the evolution of the epidemic states $x(t)$, the \phil{average weighted [weighted average of the?]} epidemic states $p^\top(t) x(t)$, and the effective reproduction number $R_o(t)$, $\forall t\geq 0$, w
% We set the initial infection level at $x(0)=0.01\times \boldsymbol{1}_n$, corresponding to $s(0)=0.99\times \boldsymbol{1}_n$, and the initial opinion states at $o(0)= \boldsymbol{0}$. We set $\beta_{\min}=0.2$, $\gamma_{\min}=0.07$, while the edge weights of the epidemic spreading network $\beta_{ij}$ are uniformly sampled from $[0.2,1]$. %Further, if $\beta_{ij}\leq \beta_{\min}$, \phil{we s}et $\beta_{ij}=\beta_{\min}+0.05$. 
% The healing rates $\gamma_i$ are uniformly sampled from $[0.07, 0.1]$. 
We set the initial condition $x(0)=0.01\times \boldsymbol{1}_n$, $s(0)=0.99\times \boldsymbol{1}_n$, and $o(0)= \boldsymbol{0}$. We also set the parameters $\beta_{\min}=0.2$, $\gamma_{\min}=0.07$, and each $\beta_{ij}$ is uniformly sampled from $[0.2,1]$. %Further, if $\beta_{ij}\leq \beta_{\min}$, \phil{we s}et $\beta_{ij}=\beta_{\min}+0.05$. 
Similarly, each $\gamma_i$ is uniformly sampled from $[0.07, 0.1]$. 
%If $\gamma_{i}\leq \gamma_{\min}$, let $\gamma_{i}=\gamma_{\min}+0.05$.
%\phil{[why not just uniformly randomly sample from $[\beta_{\min},1]$ and $[\gamma_{i},1]$, respectively?]}
%
We apply only unit edge weights to the opinion graph in all simulations.

% The behaviors of the system in \eqref{eq:S-O-C}-\eqref{eq:O-I-C} are captured by Fig. \ref{simulation1}. 
Fig. \ref{simulation1}(a) shows that the proportion of the susceptible population in all communities decreases monotonically as claimed in Lemma~\ref{lem:decreasing}, and Fig. \ref{simulation1}(b) shows the evolution of the epidemic states, with the weighted average state $x_w(t)=p^\top(t_p)x(t)$ $\forall t\geq0$ being captured by the dashed line. Note that we use $p^\top(t_p)$ for the entire time interval. Furthermore, the trend of the weighted average of the epidemic states follows the changes of the effective reproduction number $R_o(t)$ in Fig. \ref{simulation1}(d): $x_w(t)=p^\top(t_p)x(t)$ is increasing when $R_o(t)>1$; $x_w(t)=p^\top(t_p)x(t)$ is decreasing when $R_o(t)<1$. Then, $x_w(t)=p^\top(t_p)x(t)$ reaches a local peak when $R_o(t)=1$. Fig. \ref{simulation1} (a)-(d)
% (b) and (d) 
illustrate the behavior of the system in \eqref{eq:S-O-C}-\eqref{eq:O-I-C} based on the effective reproduction number $R_o(t)$ and the peak infection time $t_p$, as we proved in Theorem~\ref{thm:1}, Corollaries~\ref{cor:c1} and \ref{cor:c2}. Additionally, Fig. \ref{simulation1}(c) shows that, 
at the beginning stages of the epidemic, when no community considers the epidemic as a threat, 
% with the decreasement of 
as the susceptible population, the beliefs in the seriousness of the epidemic will increase. Meanwhile, 
% with the decreasement of 
as
the susceptible population decreases and 
% increment of 
the opinion states increase, the effective reproduction number $R_o(t)$ decreases, which aligns with Proposition \ref{prop:Spe_r}. As stated in Theorem \ref{thm:2}, %and Lemma \ref{lem:Equi}, 
the states of the system converge to zero exponentially fast. %and equal to zero after reaching the equilibrium. The opinion states at the equilibrium are determined by the susceptible states at the equilibrium.

\begin{figure}
\begin{centering}
\includegraphics[width=\columnwidth]{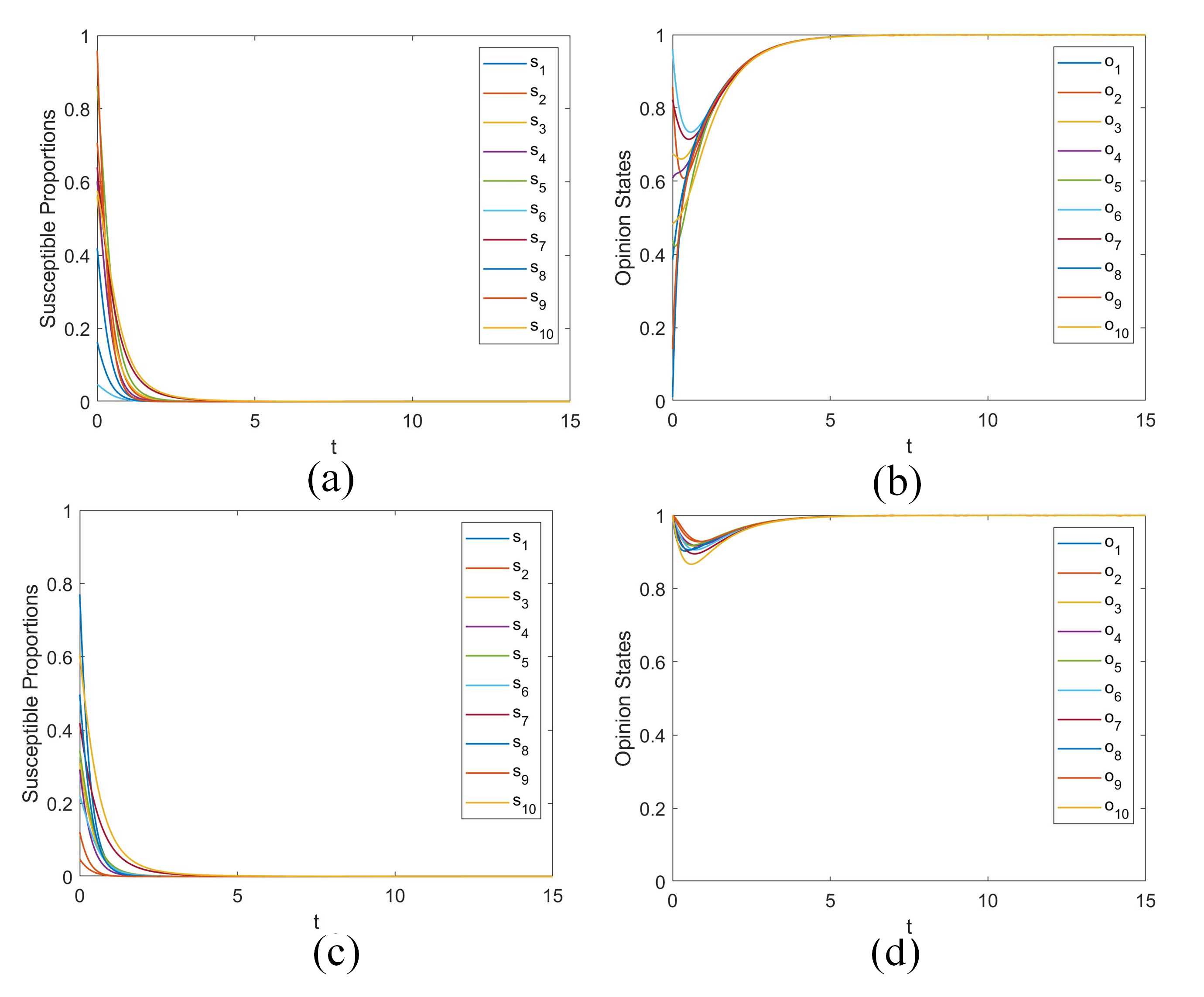}
\par\end{centering}
\centering{}\caption{States convergence with wide-spread initial opinions}
\vspace{-2ex}
\label{simulation2}
\end{figure}

Next, we will show the special case 
% that 
where 
the opinions reach consensus. As mentioned in Remark 1, the opinion states will reach consensus at the equilibrium if and only if all the communities have the same infection level. %i.e., all the communities have the same proportions of infected population at the the equilibrium. 
% which is %To let all the communities have the same proportions of infected population is 
% difficult to achieve in reality. 
We set $\beta_{\min}=0.1$ and $\gamma_{\min}=0.14$, while %the edge weights of the epidemic spreading network 
$\beta_{ij}$ and $\gamma_i$ are uniformly sampled from $[0.1,0.6]$ and $[0.14, 0.30]$, respectively to generate plots in Fig. \ref{simulation2}.
% For the opinion spreading graph in the simulation section, we only consider unit edge weights.
% Thus, in this simulation, we consider the special case 
% % that 
% where
% all the populations 
% % are 
% within each community are eventually completely infected, 
% % within each community, 
% captured by 
% Fig. \ref{simulation2}. 
% Fig. \ref{simulation2}(a)-(b) and Fig. \ref{simulation2}(c)-(d) consider two epidemic spreading processes with different initial opinion states. 
In Fig. \ref{simulation2}(a) and (c), the initial conditions of the epidemic states are uniformly sampled from $[0,1]$. In Fig. \ref{simulation2}(b) we randomly sample the initial opinion states from $[0,1]$, and In Fig. \ref{simulation2}(d)
we set the initial opinion states as $o(0)=\boldsymbol{1}_n$. 
Both Fig. \ref{simulation2}(a) and (c) 
% show the two cases that 
capture the extreme case where everyone in the population 
% became
becomes
% the populations in all communities are fully 
infected, i.e., where the susceptible states converge to $s_e=\boldsymbol{0}$. Therefore, based on Lemma \ref{consensus}, the opinion states at the equilibrium will take the form $o_e= \boldsymbol{1}_n- s_e=\boldsymbol{1}_n$, captured by Fig. \ref{simulation2}(b) and (d). The simulations demonstrate that, when reaching agreement after the 
% epidemics disappear, 
epidemic dies out,
the evaluations on the seriousness of the epidemic can reflect the infection level.
% epidemics. ....... `the epidemic' or `epidemics'

% We will use the last simulation to 
Fig. \ref{S3} aims to show that the effective reproduction number $R_o(t)$ %\phil{[shouldn't there be a subscript on this?]} 
may not decrease monotonically, unlike the classical SIR model, as stated before.
We set $\beta_{\min}=0.01$ and $\gamma_{\min}=0.05$, while $\beta_{ij}$ and $\gamma_i$ are uniformly sampled from $[0.01,0.4]$ and $[0.05, 0.1]$, respectively. 
We assume initial opinions  $o(0)=\boldsymbol{1}_n$ as shown in Fig. \ref{S3}(c), and we sample the initial infected proportion for each community from $[0.3,0.6]$ randomly.
% From 
In Fig.~\ref{S3}(d), 
% the effective reproduction number
since $R_{\min}(0)<1$, the weighted average $p^\top(t_p) x(t)$ decrease at the beginning stages of the outbreak. However, 
% during the evolution of the epidemic, 
the communities soon realize the epidemic is not as severe as they have evaluated as captured in Fig. \ref{S3}(c).
% The rebounce of their opinion states is captured in Fig. \ref{S3}(c). 
% With the decreasing of 
As the opinion states decrease, $R_o(t)$ increases, causing the weighted average $p^\top(t_p) x(t)$ to increase again, captured by Fig. \ref{S3}(b)-(d). 
% (c) and (d). 
% From 
In
Fig.~\ref{S3}(d), we observe that there are two peak candidates where 
% two moments that 
$R_o(t)=1$; we rule out the first by Theorem~\ref{thm:1},
% Definition~\ref{def:PT}, 
which states that peak infection time must satisfy 
$R_{o}(t)>1$, for $t\in [t_0, t_p)$ and $R_{o}(t)<1$, for $t \in (t_p, t_1]$, 
% with
for $t_p-t_0>0$ and  $t_1-t_p>0$ 
% being 
sufficiently small.
% $R(t_p-\tau)>1$ and $R(t_p+\tau)<1$. %[should these be $t_p$'s?]}, 
% which can be shown by Fig. \ref{S3}(b) and Fig. \ref{S3}(d). 
% Therefore, 
% from
% in
% Fig. \ref{S3}(d), 
However, 
the second 
point where
% moment when 
$R_o(t)=1$ is a peak infection time,
% but the first is not.
% This is confirmed by a comparison between 
consistent with
Fig. \ref{S3}(b) and (d).
Lastly, Fig. \ref{S3}(d) illustrates that $R_o(t)$ is not monotonically decreasing, compared to the effective reproduction number in the classical SIR model~\cite{van2002reproduction}.

%\phil{[can there really be two peak infection times?]}
%\baike{In the current model, maybe not, since the opinion states converges too fast and then follows the trend of the epidemics. However, if we can control the opinion states, there can be more than one peak infection time.}
\begin{figure}
\begin{centering}
\includegraphics[width=\columnwidth]{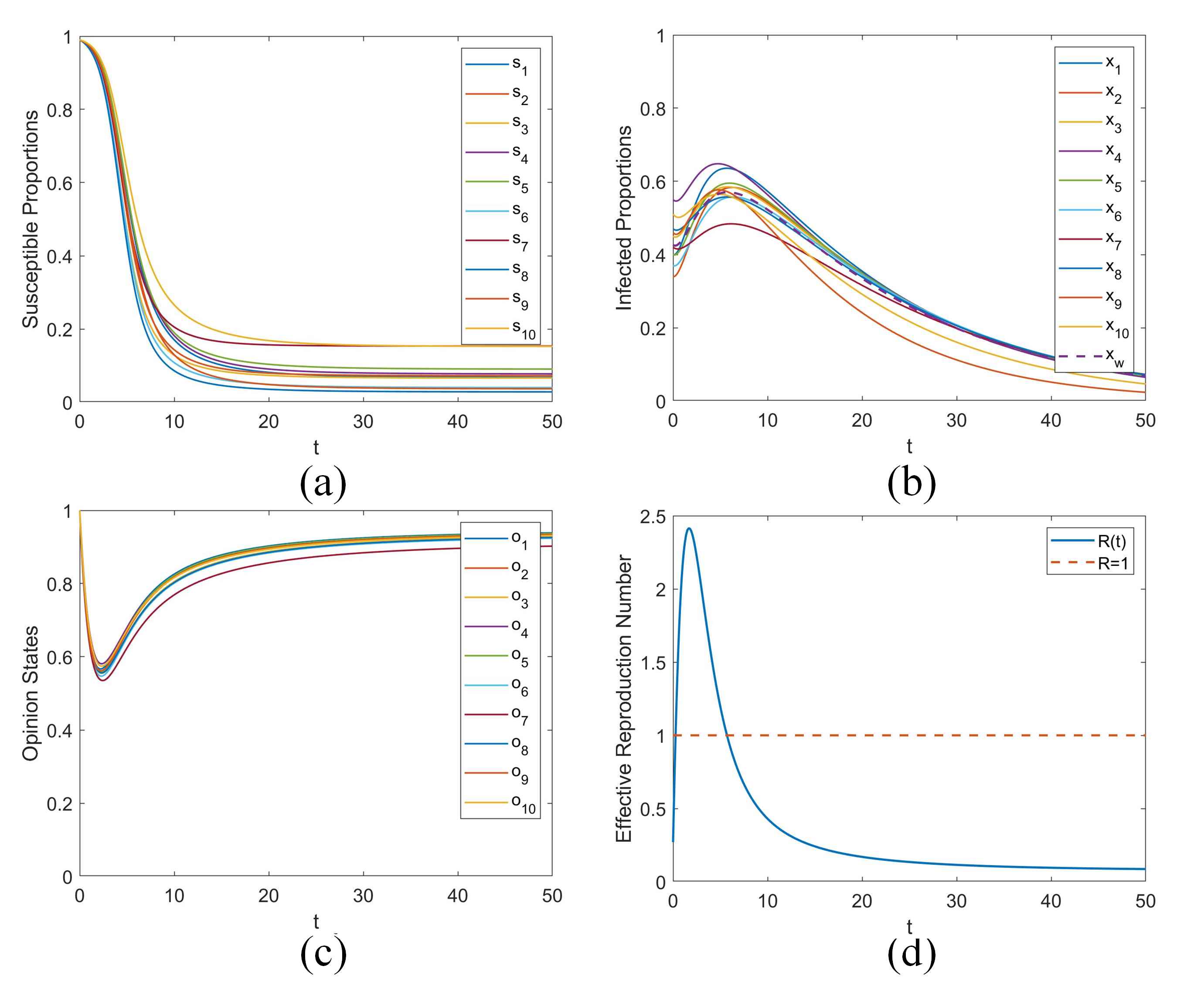}
\par\end{centering}
\centering{}\caption{State dynamics with rebouncing $R_o(t)$}
\vspace{-2ex}
\label{S3}
\end{figure}

\vspace{-1.5ex}

\section{Conclusion}

\vspace{-1ex}

In this work, we develop a networked SIR model coupled with opinion dynamics to study 
% the 
epidemic spreading processes over 
% the connected 
multiple
communities. We 
% defined
define
the effective reproduction number and peak infection time %infection thresholds 
to characterize 
the
transient behavior of the epidemic. We also 
study 
the convergence time to the equilibria. Additionally, we discover that the opinion states at the equilibria can reflect the infection level of each community to some degree. The current work can be further extended to study the influence of the structures of the opinion spreading networks on the behavior of the system. Another potential future research direction is to
% To 
design control algorithms that influence the opinions
% on opinion spreading 
to change the behavior of the epidemic.
% s is another potential research direction. 
% \baike{Additionally, we claim that there might be more than one peak infection time in Corollary \ref{cor:c2}. In future work, we will explore the condition where there will be more than one peak infection time and how to achieve that.}
% \phil{I don't think we need that last part so I commented it out.}

\vspace{-1.5ex}

\normalem
\bibliographystyle{IEEEtran}
%{\footnotesize
\bibliography{IEEEabrv,main}

%\bibliographystyle{IEEEtran}
%\bibliography{IEEEabrv,main}

\end{document}